\documentclass[pra, onecolumn, superscriptaddress, nofootinbib]{revtex4-2}

\usepackage{amsmath, amsfonts, amsthm}
\usepackage{physics, braket}
\usepackage{graphicx}
\usepackage{xcolor}
\usepackage{tikz}
\usetikzlibrary{arrows.meta,decorations.pathreplacing}
\usepackage{booktabs}
\usepackage{hyperref}
\usepackage{orcidlink}
\usepackage[capitalize]{cleveref}

\theoremstyle{plain}
\newtheorem{theorem}{Theorem}[section]
\newtheorem{proposition}[theorem]{Proposition}

\newtheorem{lemma}[theorem]{Lemma}
\newtheorem{corollary}[theorem]{Corollary}

\theoremstyle{definition}
\newtheorem{definition}[theorem]{Definition}

\numberwithin{equation}{section}

\hypersetup{
    colorlinks=true,
    linkcolor={rgb,255:red,0;green,92;blue,185},
    filecolor=magenta,      
    urlcolor=teal,
    citecolor={rgb,255:red,0;green,140;blue,0},
    pdfpagemode=UseNone,
    pdfstartview=FitH
}

\usepackage[a4paper, left=1in, right=1in, top=1in, bottom=1in]{geometry}
\usepackage{quantikz}

\begin{document}

\title{Estimation of Quantum Umlaut Information}

\author{Beomjoon Kim\,\orcidlink{0009-0007-5473-3487}}
\email{beomjoon01@snu.ac.kr}
\affiliation{Department of Psychology, Seoul National University, Seoul 08826, Korea}
\affiliation{Team QST, Seoul National University, Seoul 08826, Korea}

\author{Donghwa Ji\,\orcidlink{0009-0008-2798-7251}}
\email{donghwa722@gmail.com}
\affiliation{College of Liberal Studies, Seoul National University, Seoul 08826, Korea}
\affiliation{Team QST, Seoul National University, Seoul 08826, Korea}

\author{Kabgyun Jeong\,\orcidlink{0000-0001-7628-7835}}
\email{kgjeong6@snu.ac.kr}
\affiliation{Institute of Computer Technology, College of Engineering, Seoul National University, Seoul 08826, Korea}
\affiliation{Team QST, Seoul National University, Seoul 08826, Korea}

\date{\today}

\begin{abstract}
Quantifying correlations is important for understanding how quantum states and channels support information processing. Quantum umlaut information $U$ connects such correlations to reliability in hypothesis testing and certain communication settings, but its known formulas require classical descriptions of the underlying states. To estimate correlations without these descriptions, we introduce measured quantum umlaut information $U^M$ by replacing Umegaki relative entropy with measured relative entropy. Its variational formulation leads to a quantum generative adversarial framework for estimating $U^M$ from measurement statistics of unknown bipartite states and classical--quantum channels. This framework optimizes reference states and measurements without reconstructing the target states or channel outputs. For classical--quantum channels, it also optimizes the input probabilities. We further analyze the relation between $U^M$ and $U$, deriving a state-dependent gap bound and characterizing equality for full-rank bipartite states. Numerical simulations show trained estimates close on average to independently computed numerical references for $U^M$ across the tested states and channels. These estimates could help assess how well quantum memories preserve correlations and how noise or drift affects input--output dependence in quantum transmitters and state-preparation devices.
\end{abstract}
\maketitle

\tableofcontents

\section{Introduction}
In communication and state discrimination, a central challenge is to extract reliable information from a limited number of observations. Communicating a message involves encoding it in a physical system and recovering it through observations of that system. State discrimination uses measurements to decide which of several possible quantum states was prepared. Both tasks rely on correlations between what is encoded or prepared and what is observed. A single outcome may nevertheless be compatible with several possibilities, so additional observations can be needed to reach a reliable decision. 

In many such tasks, independent repetitions allow the error being minimized to decrease exponentially with the number of observations~\cite{Gallager1965,Audenaert2007}. The corresponding error exponent describes the asymptotic rate of that decrease: a larger exponent means that a given small error can be reached with fewer repetitions at leading order. This motivates asking how the correlations available in a task relate to its best achievable error exponent.

To test whether a device really produces the intended correlations between two quantum systems, one can compare its target preparation with independent preparations of those systems. The independent alternatives cover every possible state of the second system while preserving the individual state of the first. Observing either system alone cannot establish correlations, since the same individual states can be prepared independently. The evidence must instead come from measurements that probe the two systems together. 

A test can make two different mistakes: it can accept an independent preparation as the intended one, or fail to recognize the intended preparation when it is present. We keep the first probability below a prescribed tolerance between zero and one for every independent alternative and ask how rapidly the second can be suppressed with additional copies. In this way, the reliability of the distinction provides a task-specific measure of correlations.

Building on classical umlaut information~\cite{GirardiClassical2025}, Girardi et al.\ introduced quantum umlaut information to quantify this distinction~\cite{GirardiQuantum2025}. For a target state $\rho_{AB}$ of two systems $A$ and $B$, it is defined as
\begin{equation}\label{eq:intro-umlaut-information} U(A;B)_\rho:=\min_{\sigma_B}D(\rho_A\otimes\sigma_B\Vert\rho_{AB}), \end{equation}
where $D$ is the Umegaki relative entropy~\cite{Umegaki1962}, $\rho_A$ is the state of $A$ alone, and $\sigma_B$ ranges over states of $B$. The product $\rho_A\otimes\sigma_B$ represents the family of independent preparations described above. For this test, $U(A;B)_\rho$ gives the optimal exponent when measurements may act jointly on the available copies. It vanishes exactly for product states.

This state quantity also extends to a classical--quantum (CQ) channel $\mathcal N:x\mapsto\rho_x$, which prepares a quantum state $\rho_x$ on $B$ for each input label $x$ in a finite set. Choosing the labels with probabilities $p=(p_x)$ and recording them in a classical register $X$ gives the joint label--output state
\begin{equation}\label{eq:intro-cq-state} \omega_{XB}^p:=\sum_x p_x\lvert x\rangle\!\langle x\rvert_X\otimes\rho_x. \end{equation}

The CQ channel information is defined by optimizing the state quantity over these input probabilities:
\begin{equation}\label{eq:intro-cq-umlaut} U(\mathcal N):=\sup_p U(X;B)_{\omega^p}, \end{equation}
where the supremum is over all probability distributions on the input labels.

$U(\mathcal N)$ also connects input--output correlations to transmission reliability. Consider sending information at a small but positive rate, measured by the information transmitted per channel use. At a fixed rate, the optimal error exponent describes how rapidly decoding errors can decrease as the number of channel uses grows. In the zero-rate limit, where this rate approaches zero from above, the optimal exponent with activated non-signalling assistance is $U(\mathcal N)$~\cite{GirardiQuantum2025}.

These operational interpretations motivate assessing the correlations of actual sources and channels even when their state descriptions are unavailable. An experimental source can provide independent copies of a joint state without revealing its density matrix. A CQ device may likewise allow chosen input labels and output measurements while the matrices $\rho_x$ remain unknown. Evaluating the known full-rank formulas for $U$ requires classical state descriptions. Tomography can supply them, but reconstructs the full states when the desired output is a correlation quantity~\cite{Haah2017}. We therefore seek to estimate the relevant correlation quantity directly from measurement access, without that preliminary reconstruction.

Variational and neural methods already estimate quantum mutual information~\cite{Shin2024}, quantum entropies and measured relative entropies~\cite{Goldfeld2024}, and quantum relative entropies~\cite{LuFang2026} from copies of unknown states. For measured relative entropies, error and copy-complexity bounds are available under assumptions on the states and parametrized models~\cite{Sreekumar2026}. Estimates of joint and marginal entropies have also been combined to evaluate bounds on the cost of exchanging quantum states~\cite{Ji2026}. For umlaut information, divergence estimation must be coupled to an additional optimization: the independent comparison state $\rho_A\otimes\sigma_B$ must itself be selected by minimizing over $\sigma_B$. A CQ estimator must also optimize the input probabilities. We develop a measurement-based variational framework for these coupled optimizations.

The operational interpretation of $U$ above allows measurements to act jointly on multiple copies. Here, we seek an objective that can be evaluated by measuring one prepared copy at a time. We therefore replace the Umegaki relative entropy in the product-state comparison by measured relative entropy $D_M$. The resulting measured quantum umlaut information is defined as
\begin{equation}\label{eq:intro-mqui-definition} U^M(A;B)_\rho:=\min_{\sigma_B}D_M(\rho_A\otimes\sigma_B\Vert\rho_{AB}). \end{equation}

For two states $\tau$ and $\omega$, measured relative entropy is defined by optimizing over a common quantum measurement~\cite{BFT2017}:
\begin{equation}\label{eq:intro-measured-relative-entropy} D_M(\tau\Vert\omega):=\sup_{\mathcal M}D_{\mathrm{KL}}\!\left(p_{\mathcal M}\Vert q_{\mathcal M}\right), \end{equation}
where $\mathcal M=\{M_y\}_y$ ranges over all quantum measurements (POVMs), $p_{\mathcal M}(y)=\Tr(\tau M_y)$ and $q_{\mathcal M}(y)=\Tr(\omega M_y)$ are their outcome probabilities, and $D_{\mathrm{KL}}$ is the classical Kullback--Leibler divergence. This optimization selects the measurement that best distinguishes the two states according to their classical outcome distributions. The exact measured quantity satisfies $U^M\le U$. For bipartite states with finite information, we bound the intrinsic gap $U-U^M$; for full-rank states, we characterize when this gap vanishes. These results clarify the approximation introduced by the measured formulation, separately from errors due to circuit restrictions, incomplete optimization, and finite measurement samples.

We implement these competing optimizations as a quantum generative adversarial network~\cite{DallaireDemers2018}. The generator prepares and varies $\sigma_B$ to minimize the objective while leaving $\rho_A$ fixed. The discriminator searches for a distinguishing measurement $\mathcal M$ through a parametrized measurement basis and classical outcome scores, using an equivalent variational representation of $D_M$~\cite{BFT2017}. It maximizes the same objective that the generator minimizes. The unrestricted game value is $U^M$, which the parametrized training seeks to approximate.

For CQ channels, the same comparison uses a single reference state shared by all input labels. This preserves the interpretation of the reference as a device whose output is independent of its input. We additionally optimize the probabilities of choosing the labels to obtain the measured channel quantity $U^M(\mathcal N)$. The exact quantity lower bounds the assisted zero-rate error exponent described above.

In numerical simulations, we assess the optimization by comparing trained objective values with independently computed numerical references for $U^M$. Comparisons with $U$ then distinguish the deviation from the measured optimum from the intrinsic gap to the original quantity. The study uses complete measurement-outcome probabilities for two-, four-, and six-qubit mixed bipartite states and CQ channels whose joint input-label--output states have the same total sizes. The trained estimates are close to the measured references on average across the tested states and channels, supporting the feasibility of the proposed framework.

The remainder of this paper is organized as follows. Section~\ref{sec:qui} presents the information measures and their properties. Section~\ref{sec:gan} develops the variational estimation framework, and Section~\ref{sec:sim} reports the numerical simulations. Section~\ref{sec:conclusion} discusses applications, limitations, and future directions.

\section{Theoretical foundations}
\label{sec:qui}
\subsection{Quantum umlaut information}
\label{subsec:qui-def}

Operationally, the quantum umlaut information quantifies how distinguishable a bipartite state is from product states that share the same marginal on $A$. We recall the state and channel definitions, properties, and operational interpretations from Girardi et al.~\cite{GirardiQuantum2025}.

Throughout this paper, $D(\cdot\Vert\cdot)$ denotes the Umegaki quantum relative entropy, and $S(\cdot)$ denotes the von Neumann entropy. All Hilbert spaces are finite-dimensional, and all logarithms are natural.

\begin{definition}[Quantum umlaut information]
\label{def:umlaut}
For a bipartite quantum state $\rho_{AB}$, the quantum umlaut information is defined as
\begin{equation}
    U(A;B)_{\rho}
    :=
    \min_{\sigma_B \in \mathcal{D}(\mathcal{H}_B)}
    D\!\left(
        \rho_A \otimes \sigma_B
        \,\middle\|\,
        \rho_{AB}
    \right),
    \label{eq:def-umlaut}
\end{equation}
where $\rho_A=\Tr_B(\rho_{AB})$ and $\mathcal{D}(\mathcal{H}_B)$ denotes the set of density operators on system $B$.
\end{definition}

The quantum umlaut information is nonnegative and vanishes if and only if $\rho_{AB}=\rho_A\otimes\rho_B$. It is also nonincreasing under local quantum channels. Moreover, $U(A;B)_\rho$ is finite for every full-rank state $\rho_{AB}$, while for general states its finiteness is characterized by a support condition. Unlike the quantum mutual information, the quantum umlaut information is generally asymmetric, i.e., $U(A;B)_\rho \neq U(B;A)_\rho$.

Despite its variational definition, the quantum umlaut information admits a closed-form expression, which also identifies the unique minimizing state $\sigma_B$.

\begin{proposition}[Closed-form expression]
\label{prop:umlaut-closed-form}
Let $\rho_{AB} \succ 0$ and define the Hermitian operator
\begin{equation}
    X_B
    :=
    \Tr_A\!\left[
        (\rho_A \otimes I_B)\log \rho_{AB}
    \right].
\end{equation}
Then,
\begin{equation}
    U(A;B)_\rho
    =
    -S(\rho_A)
    -
    \log \Tr\!\left[e^{X_B}\right].
    \label{eq:umlaut-closed-form}
\end{equation}
Moreover, the unique minimizing state is
\begin{equation}
    \sigma_B^\star
    =
    \frac{e^{X_B}}{\Tr[e^{X_B}]},
    \label{eq:umlaut-marginal}
\end{equation}
which is referred to as the umlaut marginal of $\rho_{AB}$.
\end{proposition}

Quantum umlaut information extends from bipartite states to quantum channels by optimizing over the channel input. Let $\mathcal N_{A\to B}$ be a quantum channel, that is, a completely positive trace-preserving map, and let $A'$ be a reference system isomorphic to $A$. For a pure state $\psi_{A'A}=\lvert\psi\rangle\langle\psi\rvert_{A'A}$, define
\begin{equation}
    \omega_{A'B}^{\psi}
    :=(\operatorname{id}_{A'}\otimes\mathcal N_{A\to B})(\psi_{A'A}),
    \qquad
    \psi_{A'}:=\Tr_A\psi_{A'A}.
    \label{eq:channel-qui-output}
\end{equation}
The reference system is retained while $A$ passes through the channel, so that $\omega_{A'}^{\psi}=\psi_{A'}$.

\begin{definition}[Quantum channel umlaut information]
\label{def:channel-umlaut}
The quantum umlaut information of $\mathcal N_{A\to B}$ is defined as
\begin{equation}
    U(\mathcal N)
    :=
    \sup_{\psi_{A'A}}
    U(A';B)_{\omega^{\psi}}
    =
    \sup_{\psi_{A'A}}
    \min_{\sigma_B\in\mathcal D(\mathcal H_B)}
    D\!\left(
        \psi_{A'}\otimes\sigma_B
        \,\middle\|\,
        \omega_{A'B}^{\psi}
    \right),
    \label{eq:def-channel-umlaut}
\end{equation}
where the supremum is taken over all pure states $\psi_{A'A}$.
\end{definition}

Consider a replacer channel $\mathcal R_\sigma(\tau_A)=\Tr(\tau_A)\sigma_B$ that maps every input state to a fixed output state $\sigma_B$. When applied to the input subsystem of $\psi_{A'A}$, it produces the product state $\psi_{A'}\otimes\sigma_B$. The state-level comparison thus extends to distinguishing $\mathcal N$ from input-independent replacer channels, with the input state optimized. The channel information is nonnegative and vanishes if and only if $\mathcal N$ is a replacer channel.

The channel umlaut information $U(\mathcal N)$ also provides a lower bound on the zero-rate reliability exponent for classical information transmission over $\mathcal N$ with assistance from activated non-signalling correlations. For general quantum channels, joint input optimization across multiple uses can yield $U(\mathcal N^{\otimes n})>nU(\mathcal N)$, so the single-use quantity need not capture the full asymptotic performance. The optimal exponent is given by the information per use in the limit of large blocks, namely the regularized quantity $U^\infty(\mathcal N):=\lim_{n\to\infty}n^{-1}U(\mathcal N^{\otimes n})$.

For a classical--quantum (CQ) channel, the input label can be retained in a classical reference register. At a fixed input distribution, any other reference conditioned on the label can be prepared by a local channel acting on this record, which cannot increase the state umlaut information. It therefore suffices to optimize the probabilities of the input labels.

\begin{proposition}[CQ channel expression]
\label{prop:cq-umlaut}
Let $\mathcal N$ be a CQ channel that maps each input symbol $x$ in a finite alphabet $\mathcal X$ to an output state $\rho_x\in\mathcal D(\mathcal H_B)$. For a distribution $p\in\Delta(\mathcal X)$, where $\Delta(\mathcal X)$ denotes the probability simplex, define the joint state of $X$ and the channel output $B$ by
\begin{equation}
    \omega_{XB}^{p}
    :=
    \sum_{x\in\mathcal X}
    p_x\lvert x\rangle\!\langle x\rvert_X\otimes\rho_x,
    \label{eq:cq-labelled-state}
\end{equation}
where $X$ is a classical register recording the input label in an orthonormal basis $\{\lvert x\rangle_X\}_{x\in\mathcal X}$.\\
Writing $\omega_X^p=\Tr_B\omega_{XB}^p$, we have
\begin{equation}
    U(\mathcal N)
    =\sup_{p\in\Delta(\mathcal X)}U(X;B)_{\omega^p}
    =\sup_{p\in\Delta(\mathcal X)}
    \min_{\sigma_B\in\mathcal D(\mathcal H_B)}
    D\!\left(
        \omega_X^p\otimes\sigma_B
        \,\middle\|\,
        \omega_{XB}^p
    \right).
    \label{eq:cq-channel-umlaut}
\end{equation}
\end{proposition}

Thus, evaluating the CQ channel information amounts to evaluating the state quantity for $\omega_{XB}^p$ and optimizing the input distribution. For each fixed $p$, the product reference $\omega_X^p\otimes\sigma_B$ is the joint state of the recorded input label and the replacer output, so the same state-discrimination interpretation applies. Moreover, the channel umlaut information is additive under tensor products of CQ channels. Consequently, $U^\infty(\mathcal N)=U(\mathcal N)$, so the CQ channel umlaut information itself equals the zero-rate reliability exponent in the assisted communication setting described above.

While Proposition~\ref{prop:umlaut-closed-form} provides an exact characterization of the quantum umlaut information, its evaluation requires a classical description of $\rho_{AB}$, in particular through the matrix logarithm $\log\rho_{AB}$. The channel expressions likewise involve matrix logarithms of the corresponding joint reference--output states. These requirements become restrictive when the target state or channel outputs are accessible only through repeated preparations and measurement outcomes.

\subsection{Measured quantum umlaut information}
\label{subsec:measured-qui}

To express the product-state comparison in terms of measurement outcomes, we replace the Umegaki relative entropy in $U$ with measured relative entropy. We first introduce the resulting state quantity and analyze its relation to $U$, then extend the construction to quantum channels.

The measured relative entropy between quantum states $\tau$ and $\omega$ is defined as~\cite{BFT2017}
\begin{equation}
    D_M(\tau\Vert\omega)
    :=
    \sup_{\mathcal{M}}
    D_{\mathrm{KL}}\!\left(
        p_{\mathcal{M}}
        \,\middle\|\,
        q_{\mathcal{M}}
    \right),
\end{equation}
where the supremum is taken over all POVMs $\mathcal{M}$, and $p_{\mathcal{M}}$ and $q_{\mathcal{M}}$ denote the corresponding measure\-ment-outcome distributions obtained from $\tau$ and $\omega$, respectively. Thus, $D_M(\tau\Vert\omega)$ quantifies the largest classical distinguishability between $\tau$ and $\omega$ that can be extracted through a quantum measurement.

Motivated by this measurement-based notion of distinguishability, we introduce the measured quantum umlaut information.

\begin{definition}[Measured quantum umlaut information]
\label{def:mqui}
For a bipartite quantum state $\rho_{AB}$, the measured quantum umlaut information is defined as
\begin{equation}
    U^M(A;B)_{\rho} := \min_{\sigma_B \in \mathcal{D}(\mathcal{H}_B)} D_M(\rho_A \otimes \sigma_B \Vert \rho_{AB}).
    \label{eq:def-mqui}
\end{equation}
\end{definition}

The measured quantum umlaut information is intended to retain the correlation-detection properties of the quantum umlaut information while restricting the distinguishability measure to information accessible through measurements. Since the measured relative entropy satisfies
$D_M(\tau\Vert\omega)\le D(\tau\Vert\omega)$, it immediately follows that
\begin{equation}
U^M(A;B)_\rho \le U(A;B)_\rho.
\end{equation}
More generally, the two quantities satisfy the following properties.

\begin{theorem}[Relation between $U^M$ and $U$]
\label{thm:mqui-properties}
For any bipartite state $\rho_{AB}$, the following properties hold. 
\begin{enumerate}
    \item $U^M(A;B)_\rho \ge 0$, with equality if and only if
    $\rho_{AB}=\rho_A\otimes\rho_B$;
    
    \item $U^M(A;B)_\rho<\infty$ if and only if
    $U(A;B)_\rho<\infty$;
    
    \item $U^M$ is invariant under local unitaries and nonincreasing under local quantum channels;
    
    \item whenever the quantities are finite,
    \begin{equation}
        U^M(A;B)_\rho
        \le
        U(A;B)_\rho
        \le
        U^M(A;B)_\rho
        +
        \log |\operatorname{spec}(\rho_{AB})|
        \label{eq:mqui-gap}
    \end{equation}
    where $|\operatorname{spec}(\rho_{AB})|$ denotes the number of distinct nonzero eigenvalues of $\rho_{AB}$.
\end{enumerate}
\end{theorem}

\begin{proof}[Proof sketch]
The first three statements follow from the corresponding properties of the measured relative entropy. For the last statement, we use the spectral pinching bound~\cite{BBH2021}
\begin{equation}
    D_M(\tau\Vert\omega)
    \le
    D(\tau\Vert\omega)
    \le
    D_M(\tau\Vert\omega)
    +
    \log |\operatorname{spec}(\omega)|,
\end{equation}
and minimize over $\sigma_B$. Since the additive term is independent of $\sigma_B$, this yields Eq.~\eqref{eq:mqui-gap}.
\end{proof}

Theorem~\ref{thm:mqui-properties} shows that $U^M$ preserves the basic correlation-detection properties of $U$ and provides a lower approximation to it.

For finite information quantities, accuracy relative to $U$ depends on two sources of error: the error in estimating $U^M$, including circuit restrictions, incomplete optimization, and finite measurement samples, and the intrinsic gap $U-U^M$. The following analysis concerns the intrinsic gap, which for bipartite states is defined as
\begin{equation}\label{eq:mqui-intrinsic-gap} \Delta_\rho:=U(A;B)_\rho-U^M(A;B)_\rho. \end{equation}
This gap arises from the choice of measured relative entropy and persists even with exact measurement statistics and unrestricted optimization. The spectral bound above is uniform, but it need not describe how close the two quantities are for a particular state. For full-rank states, the optimizing product reference for $U$ provides a sharper characterization.

\begin{proposition}[State-dependent gap and equality condition]
\label{prop:mqui-state-gap}
Let $\rho_{AB}\succ0$, and let $\sigma_B^\star$ be its umlaut marginal. Define the optimal product reference $\tau_{AB}^\star=\rho_A\otimes\sigma_B^\star$ and its logarithmic difference from the target state by
\begin{equation}\label{eq:mqui-gap-witness} H^\star:=\log\tau_{AB}^\star-\log\rho_{AB}. \end{equation}
The state-dependent quantity
\begin{equation}\label{eq:mqui-gap-gamma} \Gamma_\rho:=\log\Tr\!\left[\rho_{AB}e^{H^\star}\right] \end{equation}
is nonnegative and bounds the intrinsic gap as
\begin{equation}\label{eq:mqui-state-gap} 0\le\Delta_\rho\le\min\!\left\{U(A;B)_\rho,\,\log|\operatorname{spec}(\rho_{AB})|,\,\Gamma_\rho\right\}. \end{equation}
Moreover,
\begin{equation}\label{eq:mqui-equality-condition} U^M(A;B)_\rho=U(A;B)_\rho\quad\Longleftrightarrow\quad[\rho_{AB},\tau_{AB}^\star]=0. \end{equation}
\end{proposition}

\begin{proof}[Proof sketch]
The Gibbs form of $\sigma_B^\star$ gives $\Tr[(\rho_A\otimes\sigma_B)H^\star]=U(A;B)_\rho$ for every $\sigma_B$. Evaluating the variational expression of measured relative entropy stated in Lemma~\ref{lem:dm-variational} at $H^\star-\Gamma_\rho I_{AB}$ therefore gives the uniform lower bound $D_M(\rho_A\otimes\sigma_B\Vert\rho_{AB})\ge U(A;B)_\rho-\Gamma_\rho$. Minimizing over $\sigma_B$ yields $\Delta_\rho\le\Gamma_\rho$, while nonnegativity and the spectral bound give the other two bounds. The Golden--Thompson inequality ensures $\Gamma_\rho\ge0$. Commutation of $\rho_{AB}$ and $\tau_{AB}^\star$ gives $\Gamma_\rho=0$, and the converse follows from the full-rank equality condition for measured relative entropy. The full proof is given in Appendix~\ref{sec:proof-mqui-state-gap}.
\end{proof}

The equality condition concerns commutation with the optimal product reference and is not restricted to uncorrelated states. For a full-rank state diagonal in a product basis, the umlaut marginal is diagonal in the corresponding basis on $B$, so $U^M=U$ even in the presence of classical correlations. More generally, the state-dependent bound quantifies when $U^M$ is a close approximation to $U$. The quantity $\Gamma_\rho$ uses $\log\rho_{AB}$ and the exact umlaut marginal, so it provides a theoretical bound and a benchmark for simulations with known target states; it is not directly supplied by the measurement statistics used to train the estimator.

The same bound controls small departures from states with zero intrinsic gap. In particular, it gives a quadratic approximation guarantee around correlated equality states without evaluating $U^M$ separately.

\begin{corollary}[Quadratic stability of the intrinsic gap]
\label{cor:mqui-quadratic-stability}
Let $\rho_{AB}(\varepsilon)\succ0$ be a twice continuously differentiable family of density operators for $\varepsilon\in(-\varepsilon_0,\varepsilon_0)$, with $\varepsilon_0>0$, and let $\tau_{AB}^\star(\varepsilon)$ denote its optimal product reference. If $[\rho_{AB}(0),\tau_{AB}^\star(0)]=0$, then there exist $0<\delta<\varepsilon_0$ and $C\ge0$ such that
\begin{equation}\label{eq:mqui-quadratic-gap} 0\le\Delta_{\rho(\varepsilon)}\le C\varepsilon^2\qquad (|\varepsilon|\le\delta). \end{equation}
If $U(A;B)_{\rho(0)}>0$, then, after reducing $\delta$ if necessary,
\begin{equation}\label{eq:mqui-quadratic-relative-gap} \frac{\Delta_{\rho(\varepsilon)}}{U(A;B)_{\rho(\varepsilon)}}\le\frac{2C}{U(A;B)_{\rho(0)}}\varepsilon^2\qquad (|\varepsilon|\le\delta). \end{equation}
\end{corollary}

\begin{proof}[Proof sketch]
The function $g(\varepsilon)=\Gamma_{\rho(\varepsilon)}$ is twice continuously differentiable, nonnegative, and satisfies $g(0)=0$. Hence $g'(0)=0$, and Taylor's theorem gives $g(\varepsilon)=O(\varepsilon^2)$. The state-dependent bound transfers this estimate to $\Delta_{\rho(\varepsilon)}$. When $U(A;B)_{\rho(0)}>0$, continuity keeps the denominator bounded away from zero and gives the relative bound. The full proof is given in Appendix~\ref{sec:proof-mqui-quadratic-stability}.
\end{proof}

A correlated starting point is provided by a two-qubit Bell-diagonal state with nonuniform weights, $\rho_{AB}(0)=\sum_{k=1}^{4}p_k|\beta_k\rangle\langle\beta_k|$, where $\{|\beta_k\rangle\}$ is the Bell basis, $p_k>0$, and $\sum_kp_k=1$. Since each Bell state has maximally mixed marginals, the Gibbs formula gives
\begin{equation}\label{eq:mqui-bell-diagonal-reference} X_B=\frac14\left(\sum_{k=1}^{4}\log p_k\right)I_B,\qquad\sigma_B^\star=\frac{I_B}{2},\qquad\tau_{AB}^\star=\frac{I_{AB}}{4}. \end{equation}
Consequently, $\Delta_{\rho(0)}=0$, whereas
\begin{equation}\label{eq:mqui-bell-diagonal-umlaut} U(A;B)_{\rho(0)}=-\log4-\frac14\sum_{k=1}^{4}\log p_k. \end{equation}
This value is positive whenever the weights are not all $1/4$. Corollary~\ref{cor:mqui-quadratic-stability} therefore gives both absolute and relative gap bounds for sufficiently small smooth full-rank perturbations of these states while their quantum umlaut information remains positive. The guarantee is local: the constant $C$ and neighborhood size $\delta$ depend on the starting state and the chosen perturbation family.

For quantum channels, we extend the same measurement-based construction from bipartite states by optimizing the state quantity over the channel input.
\begin{definition}[Measured quantum channel umlaut information]
\label{def:channel-mqui}
The measured quantum umlaut information of a channel $\mathcal N_{A\to B}$ is defined as
\begin{equation}
    U^M(\mathcal N)
    :=\sup_{\psi_{A'A}}U^M(A';B)_{\omega^\psi}
    =\sup_{\psi_{A'A}}\min_{\sigma_B\in\mathcal D(\mathcal H_B)}
    D_M\!\left(\psi_{A'}\otimes\sigma_B\,\middle\|\,\omega_{A'B}^\psi\right),
    \label{eq:def-channel-mqui}
\end{equation}
where the supremum is over pure reference--input states and $\omega_{A'B}^\psi$ is obtained by applying $\mathcal N$ to $A$ while retaining $A'$ as a reference.
\end{definition}

Taking the supremum of the state-level inequality $U^M\le U$ over the channel input and using the relation $U(\mathcal N)\le U^\infty(\mathcal N)$ gives
\begin{equation}
    0\le U^M(\mathcal N)\le U(\mathcal N)\le U^\infty(\mathcal N).
    \label{eq:channel-mqui-bound}
\end{equation}
The first statement of Theorem~\ref{thm:mqui-properties} implies that $U^M(\mathcal N)=0$ if and only if $U(\mathcal N)=0$, and hence if and only if $\mathcal N$ is a replacer channel.

The measured quantum umlaut information of a CQ channel can likewise be obtained by reducing the input optimization to one over probability distributions $p$ on the classical input labels.

\begin{proposition}[Measured CQ channel expression]
\label{prop:cq-mqui}
For a CQ channel $\mathcal N:x\mapsto\rho_x$ and an input distribution $p$, let $\omega_{XB}^p$ be the joint state in which the classical register $X$ records the input label $x$ with probability $p_x$ and $B$ is in the corresponding output state $\rho_x$. Then
\begin{equation}
    U^M(\mathcal N)
    =\sup_{p\in\Delta(\mathcal X)}U^M(X;B)_{\omega^p}
    =\sup_{p\in\Delta(\mathcal X)}\min_{\sigma_B\in\mathcal D(\mathcal H_B)}
    D_M\!\left(\omega_X^p\otimes\sigma_B\,\middle\|\,\omega_{XB}^p\right).
    \label{eq:cq-channel-mqui}
\end{equation}
\end{proposition}

\begin{proof}
Write an arbitrary pure reference--input state as $|\psi\rangle_{A'A}=\sum_x\sqrt{p_x}|\phi_x\rangle_{A'}|x\rangle_A$, where the reference vectors $|\phi_x\rangle$ are normalized and may be chosen arbitrarily when $p_x=0$. Consider the channel that measures $X$ in the label basis and prepares the corresponding reference state:
\begin{equation}\label{eq:cq-mqui-reference-channel} \Lambda_{X\to A'}(T):=\sum_x\langle x|T|x\rangle\,|\phi_x\rangle\langle\phi_x|. \end{equation}
The CQ channel action gives $\omega_{A'B}^{\psi}=(\Lambda_{X\to A'}\otimes\operatorname{id}_B)(\omega_{XB}^p)$. The third statement of Theorem~\ref{thm:mqui-properties} therefore implies
\begin{equation}\label{eq:cq-mqui-local-bound} U^M(A';B)_{\omega^{\psi}}\le U^M(X;B)_{\omega^p}. \end{equation}
Taking the supremum over pure reference--input states gives $U^M(\mathcal N)\le\sup_p U^M(X;B)_{\omega^p}$.

Conversely, for every input distribution $p$, the pure state $|\psi^p\rangle_{A'A}=\sum_x\sqrt{p_x}|x\rangle_{A'}|x\rangle_A$ yields exactly $\omega_{XB}^p$ after applying $\mathcal N$ to $A$, with $A'$ identified with $X$. The input optimization thus gives $U^M(\mathcal N)\ge U^M(X;B)_{\omega^p}$ for every $p$, proving the first equality. The second equality follows from the definition of $U^M$ for bipartite states.
\end{proof}

The minimization uses a common reference state $\sigma_B$ for all input labels, while the distribution $p$ is optimized to obtain the channel quantity. Since $U^\infty(\mathcal N)=U(\mathcal N)$ for CQ channels, the exact measured quantum umlaut information $U^M(\mathcal N)$ is a lower bound on the zero-rate reliability exponent in the assisted communication setting described in Sec.~\ref{subsec:qui-def}. When these quantities are finite, $U(\mathcal N)-U^M(\mathcal N)$ is the intrinsic approximation gap to this exponent.

\subsection{Quantum generative adversarial framework}
\label{subsec:qgan}

The minimization over the reference state and the maximization over measurements can be assigned to two competing learning models. A generative adversarial network (GAN) consists of two models trained against each other: a generator, which produces samples intended to reproduce a target data distribution, and a discriminator, which attempts to distinguish the generated samples from genuine ones~\cite{Goodfellow2014}. Given a generator $G_\theta$ and a discriminator $D_\phi$, the standard GAN is formulated as the minimax problem
\begin{equation}
    \min_{\theta}\max_{\phi} V(\theta,\phi),
\end{equation}
where
\begin{equation}
    V(\theta,\phi)
    =
    \mathbb{E}_{x\sim p_{\mathrm{data}}}
    \!\left[\log D_\phi(x)\right]
    +
    \mathbb{E}_{z\sim p_z}
    \!\left[
        \log\!\bigl(1-D_\phi(G_\theta(z))\bigr)
    \right].
    \label{eq:gan}
\end{equation}
Here, $p_{\mathrm{data}}$ denotes the target distribution, $G_\theta$ maps a random input $z\sim p_z$ to a generated sample, and $D_\phi(x)\in[0,1]$ estimates the probability that $x$ is genuine. The discriminator maximizes the objective by distinguishing genuine and generated samples, whereas the generator minimizes it by making the generated distribution increasingly difficult to distinguish from the target distribution.

Quantum generative adversarial networks (QGANs) extend this adversarial framework to quantum data~\cite{DallaireDemers2018}. In a typical QGAN, the generator is implemented by a parametrized quantum circuit that prepares a quantum state $\sigma_\theta$, while the discriminator implements a quantum measurement designed to distinguish $\sigma_\theta$ from a target state. The two components are trained adversarially: the discriminator seeks measurements that better distinguish the two states, while the generator adjusts its parameters to reduce their distinguishability. In the idealized setting with unrestricted state preparation and measurements, the generated state coincides with the target state at equilibrium, where the discriminator cannot distinguish them. Information measures can also guide quantum generative training. For example, InfoQGAN uses a classical neural estimator to maximize the mutual information between latent codes and generated classical samples~\cite{Lee2025InfoQGAN}.

The connection to our estimation problem comes from variational representations of divergences. In the classical setting, $f$-GANs use these representations to construct adversarial objectives~\cite{Nowozin2016}. Their KL formulation uses real-valued discriminator scores and is the classical counterpart of our objective for measured relative entropy.

Here, the generator minimizes the objective over $\sigma_B$, while the discriminator maximizes it by distinguishing $\rho_A\otimes\sigma_B$ from $\rho_{AB}$. The unrestricted game value is $U^M(A;B)_\rho$. Section~\ref{sec:gan} develops this objective for bipartite states and CQ channels.

\section{Estimation of quantum umlaut information}
\label{sec:gan}

In this section, we develop a variational framework for estimating the measured quantum umlaut information of bipartite states, $U^M(A;B)_\rho$, and CQ channels, $U^M(\mathcal N)$, from quantum measurements. We focus on targets for which $U$ is finite, so that $U^M$ is also finite. These include full-rank bipartite states and CQ channels with full-rank output states. Our approach combines a variational characterization of the measured relative entropy with the quantum generative adversarial framework introduced in Sec.~\ref{subsec:qgan}. We first derive an adversarial min--sup formulation for bipartite states and the corresponding variational representation for CQ channels, which includes an outer optimization over the classical input distribution. We then show how the optimization variables are parametrized and how the resulting objectives are estimated from measurement outcomes and optimized through alternating updates of the variational parameters.

\subsection{Variational formulation}
\label{subsec:variational-formulation}

We begin by rewriting the measured relative entropy in a form suitable for variational estimation. Although $D_M$ is defined through an optimization over quantum measurements, it admits an equivalent characterization as an optimization over Hermitian operators, which serves as the basis of our adversarial formulation.

\begin{lemma}[Variational expression of measured relative entropy~\cite{BFT2017}]
\label{lem:dm-variational}
For quantum states $\tau$ and $\omega$,
\begin{equation}
    D_M(\tau\Vert\omega)
    =
    \sup_{H=H^\dagger}
    \left\{
        \Tr(\tau H)
        -
        \Tr(\omega e^H)
        +
        1
    \right\},
    \label{eq:dm-variational}
\end{equation}
where the supremum is taken over all Hermitian operators $H$ acting on the underlying Hilbert space.
\end{lemma}

Substituting the variational representation of Lemma~\ref{lem:dm-variational} into the definition of $U^M(A;B)_\rho$, with $\tau=\rho_A\otimes\sigma_B$ and $\omega=\rho_{AB}$, gives
\begin{align}
    U^M(A;B)_\rho
    &=
    \min_{\sigma_B\in\mathcal{D}(\mathcal{H}_B)}
    \sup_{H=H^\dagger}
    \left\{
        \Tr\!\left[(\rho_A\otimes\sigma_B)H\right]
        -
        \Tr\!\left[\rho_{AB}e^H\right]
        +
        1
    \right\}.
    \label{eq:mqui-minsup}
\end{align}
Thus, the measured quantum umlaut information of a bipartite state can be expressed as an explicit min--sup optimization over a quantum state $\sigma_B$ and a Hermitian operator $H$. For convenience, we define
\begin{equation}
    \mathcal{L}(\sigma_B,H)
    :=
    \Tr\!\left[(\rho_A\otimes\sigma_B)H\right]
    -
    \Tr\!\left[\rho_{AB}e^H\right]
    +
    1.
    \label{eq:mqui-objective}
\end{equation}
Once $\sigma_B$ and $H$ are suitably parametrized, the two trace terms in this objective can be evaluated from quantum measurement statistics, making the min--sup formulation amenable to variational estimation.

This min--sup formulation admits a natural adversarial interpretation. The state $\sigma_B$ plays the role of the generator variable: it specifies a product state $\rho_A\otimes\sigma_B$ and is optimized to minimize the objective. In contrast, the Hermitian operator $H$ plays the role of the discriminator variable and is optimized to maximize the distinguishability between $\rho_A\otimes\sigma_B$ and the target state $\rho_{AB}$ captured by the variational objective.

This game differs from the standard QGAN setting in an important respect. The generator is restricted to the family of product states $\rho_A\otimes\sigma_B$ and therefore cannot, in general, reproduce a correlated target state $\rho_{AB}$. Consequently, the quantity of interest is not the generated state itself, but the optimal value of the adversarial game, which is precisely $U^M(A;B)_\rho$. The optimizing state $\sigma_B$ is instead a minimizer in the definition of $U^M$.

The generator optimization is essential in this setting. Unlike variational formulations of the quantum mutual information, where the optimal product reference is determined by the marginal state~\cite{GirardiQuantum2025}, the minimizer of the reverse relative-entropy problem is generally nontrivial. For the quantum umlaut information of a full-rank state $\rho_{AB}$, Proposition~\ref{prop:umlaut-closed-form} identifies this optimizer as the umlaut marginal, which depends on the global state $\rho_{AB}$. For the measured quantum umlaut information, no analogous closed-form minimizer is available in general. Consequently, both the generator and the discriminator must be optimized.

The variational formulation for quantum channels adds an outer optimization over the channel input to the min--sup problem for bipartite states. For channel estimation, we focus on CQ channels $\mathcal N:x\mapsto\rho_x$, for which the outer optimization is over classical input distributions. Combining Proposition~\ref{prop:cq-mqui} with Lemma~\ref{lem:dm-variational} gives
\begin{equation}
    U^M(\mathcal N)
    =
    \sup_{p\in\Delta(\mathcal X)}
    \min_{\sigma_B\in\mathcal D(\mathcal H_B)}
    \sup_{H=H^\dagger}
    \left\{
        \Tr\!\left[(\omega_X^p\otimes\sigma_B)H\right]
        -\Tr\!\left[\omega_{XB}^p e^H\right]
        +1
    \right\}.
    \label{eq:cq-channel-mqui-variational}
\end{equation}
Here, $H$ acts on $XB$, and the distribution $p$ enters both the product reference state $\omega_X^p\otimes\sigma_B$ and the labelled output state $\omega_{XB}^p$. For each fixed $p$, the inner min--sup problem gives $U^M(X;B)_{\omega^p}$, with a common generator state $\sigma_B$ for all input labels.

To turn these abstract optimization problems into trainable quantum estimators, we next describe how to parametrize the generator state $\sigma_B$ and the discriminator operator $H$, together with the classical input distribution $p$ for CQ channels.

\subsection{Methodology}
\label{subsec:implementation}

The generator prepares a quantum state $\sigma_B(\theta)$ on $B$ via a
parametrized purification on $BR$:
\begin{equation}
    \ket{\Psi(\theta)}_{BR}
    =
    U_{BR}(\theta)\ket{0}^{\otimes(n_B+n_R)},
    \qquad
    \sigma_B(\theta)
    =
    \Tr_R\!\left[
        \ket{\Psi(\theta)}\!\bra{\Psi(\theta)}
    \right],
    \label{eq:generator}
\end{equation}
where $R$ is an ancillary register, $U_{BR}(\theta)$ is a parametrized
unitary circuit, and $n_B$ and $n_R$ denote the numbers of qubits in $B$
and $R$, respectively. The register $R$ is left unmeasured in the sampling
circuits shown in Fig.~\ref{fig:framework}, and $n_R$ is treated as a
tunable hyperparameter controlling the maximum rank of $\sigma_B(\theta)$.
In particular,
\begin{equation}
    \operatorname{rank}\sigma_B(\theta)
    \le \min\{2^{n_B},2^{n_R}\}.
    \label{eq:generator-rank}
\end{equation}
Representing arbitrary states on $B$ therefore requires $n_R\ge n_B$
and a sufficiently expressive preparation circuit; a smaller register or
an insufficiently expressive circuit restricts the generator family.

For the discriminator, we append an $n_D$-qubit ancillary register $D$
initialized in $\ket{0}_D := \ket{0}^{\otimes n_D}$, where $n_D$ is treated
as a tunable hyperparameter. The register $D$ provides additional degrees of freedom for implementing measurements with a restricted discriminator circuit. Adjoining the same fixed ancilla to both
arguments leaves the measured relative entropy invariant by data processing.
We therefore consider the extended states
\begin{equation}
    \widetilde{\tau}_{\theta}
    :=
    \left(\rho_A\otimes\sigma_B(\theta)\right)
    \otimes\ket{0}\!\bra{0}_D,
    \qquad
    \widetilde{\rho}
    :=
    \rho_{AB}\otimes\ket{0}\!\bra{0}_D.
    \label{eq:extended-states}
\end{equation}

We parametrize the discriminator observable $H(h,\phi)$ through its eigenbasis and eigenvalues~\cite{Goldfeld2024}:
\begin{equation}
    H(h,\phi)
    =
    \sum_{i=0}^{d-1}
    h_i\,V(\phi)\ket{i}\!\bra{i}V^\dagger(\phi),
    \qquad
    d = 2^{n_A+n_B+n_D},
    \label{eq:discriminator}
\end{equation}
where $V(\phi)$ is a parametrized unitary circuit on $ABD$ that determines the eigenbasis, $n_A$ denotes the number of qubits in $A$, and $h = (h_0,\ldots,h_{d-1}) \in \mathbb{R}^{d}$ contains the trainable eigenvalues. 

Since exponentiation preserves the eigenbasis of $H(h,\phi)$ and replaces
each eigenvalue $h_i$ by $e^{h_i}$, both terms of the variational objective
can be evaluated from computational-basis measurement statistics.
Specifically, applying $V^\dagger(\phi)$ and measuring all qubits in $ABD$
yields
\begin{equation}
    p_{\theta,\phi}(i) := \bra{i}V^\dagger(\phi)\,\widetilde{\tau}_{\theta}\,V(\phi)\ket{i},
    \qquad
    q_{\phi}(i) := \bra{i}V^\dagger(\phi)\,\widetilde{\rho}\,V(\phi)\ket{i},
    \qquad
    i = 0,\ldots,d-1.
    \label{eq:pq}
\end{equation}
The same discriminator circuit and measurement are used in both cases,
with only the input state changed. In terms of these distributions, the
variational loss becomes
\begin{equation}
    \mathcal{L}(\theta,h,\phi)
    =
    \sum_{i=0}^{d-1}h_i\,p_{\theta,\phi}(i)
    -
    \sum_{i=0}^{d-1}e^{h_i}\,q_{\phi}(i)
    +1.
    \label{eq:loss}
\end{equation}

We apply the preceding construction to a CQ channel $\mathcal N:x\mapsto\rho_x$ at a fixed input distribution $p\in\Delta(\mathcal X)$, replacing $A$, $\rho_{AB}$, and $\rho_A$ by $X$, $\omega_{XB}^p$, and $\omega_X^p$, respectively. The generator and discriminator retain the structures used for state estimation. The generator prepares a common state $\sigma_B(\theta)$ for all input labels via a parametrized purification on $BR$, while the discriminator uses the spectral parametrization of $H(h,\phi)$, with $V(\phi)$ acting on $XBD$. The ancillary registers $R$ and $D$ are used as above, with their sizes $n_R$ and $n_D$ remaining tunable hyperparameters.

The corresponding measurement-outcome distributions are
\begin{equation}
    \begin{aligned}
        P_{\theta,\phi}^{p}(i)
        &:=
        \bra{i}V^\dagger(\phi)
        \bigl[
            (\omega_X^p\otimes\sigma_B(\theta))
            \otimes\ket{0}\!\bra{0}_D
        \bigr]
        V(\phi)\ket{i},\\
        Q_{\phi}^{p}(i)
        &:=
        \bra{i}V^\dagger(\phi)
        \bigl[
            \omega_{XB}^{p}\otimes\ket{0}\!\bra{0}_D
        \bigr]
        V(\phi)\ket{i},
        \qquad i=0,\ldots,d_{\mathrm{CQ}}-1,
    \end{aligned}
    \label{eq:cq-measurement-distributions}
\end{equation}
where $i$ runs over the computational-basis outcomes on $XBD$, whose dimension is $d_{\mathrm{CQ}}=2^{n_X+n_B+n_D}$. Here, $n_X$ denotes the number of qubits encoding the classical label register $X$. Substituting these distributions into Eq.~\eqref{eq:loss} gives the CQ loss for fixed $p$
\begin{equation}
    \mathcal L_{\mathrm{CQ}}(p;\theta,h,\phi)
    =
    \sum_i h_iP_{\theta,\phi}^{p}(i)
    -
    \sum_i e^{h_i}Q_{\phi}^{p}(i)
    +1.
    \label{eq:cq-variational-loss}
\end{equation}

For restricted discriminator circuits, Lemma~\ref{lem:dm-variational}
and the ancilla invariance  give

\begin{equation}
    \begin{aligned}
        \sup_{h,\phi}\mathcal{L}(\theta,h,\phi)
        &\le D_M(\rho_A\otimes\sigma_B(\theta)\Vert\rho_{AB}),\\
        \sup_{h,\phi}\mathcal L_{\mathrm{CQ}}(p;\theta,h,\phi)
        &\le D_M(\omega_X^p\otimes\sigma_B(\theta)\Vert\omega_{XB}^p).
    \end{aligned}
    \label{eq:restricted-discriminator}
\end{equation}

These bounds hold at the current generated state $\sigma_B(\theta)$, with $p$ held fixed in the CQ case. Equality is recovered in each case when the discriminator can represent arbitrary Hermitian operators.  Circuit restrictions and incomplete optimization do not guarantee that the learned losses provide lower or upper bounds on the corresponding information quantities $U^M(A;B)_\rho$ and $U^M(\mathcal N)$.

\begin{figure}[!htbp]
    \definecolor{figGenerator}{HTML}{004488}
    \definecolor{figDiscriminator}{HTML}{BB5566}
    \tikzset{
        figone wire/.style={line width=0.45pt},
        figone block/.style={draw, line width=0.55pt, fill=white, inner sep=3pt},
        figone generator/.style={figone block,draw=figGenerator,fill=figGenerator!22},
        figone discriminator/.style={figone block,draw=figDiscriminator,fill=figDiscriminator!22},
        pics/figone discard/.style={code={
            \draw[figone wire] (0,-0.13) -- (0,0.13);
            \draw[figone wire] (0.08,-0.09) -- (0.08,0.09);
            \draw[figone wire] (0.16,-0.045) -- (0.16,0.045);
        }},
        pics/figone readout/.style={code={
            \draw[line width=0.45pt,fill=white] (-0.25,-0.2) rectangle (0.25,0.2);
            \draw[line width=0.4pt] (-0.16,-0.1) arc[start angle=180,end angle=0,radius=0.16];
            \draw[line width=0.4pt,-{Latex[length=1mm,width=0.7mm]}] (0,-0.1) -- (0.12,0.12);
        }}
    }
    \centering
    \begin{minipage}[t]{0.60\linewidth}
    \centering
    \begin{tikzpicture}[x=1cm,y=1cm,font=\normalsize]
        \path[use as bounding box] (-0.65,-2.9) rectangle (6.6,0.4);

        \draw[figone wire,decorate,decoration={brace,amplitude=3pt,mirror}] (0.35,0.12) -- (0.35,-0.72);
        \node[anchor=east] at (0.16,-0.3) {$\rho_{AB}$};
        \node[anchor=east] (inputB) at (0.55,-1.2) {$\ket{0}^{\otimes n_B}$};
        \node[anchor=east] (inputR) at (0.55,-1.8) {$\ket{0}^{\otimes n_R}$};
        \node[anchor=east] (inputD) at (0.55,-2.4) {$\ket{0}_D$};
        \node[figone generator,minimum width=1.4cm,minimum height=1.1cm]
            (generator) at (2.2,-1.5) {$U_{BR}(\theta)$};
        \node[figone discriminator,minimum width=1cm,minimum height=2.95cm]
            (discP) at (4.7,-1.2) {$V^\dagger(\phi)$};

        \draw[figone wire] (0.55,0) -- (discP.west |- 0,0);
        \draw[figone wire] (0.55,-0.6) -- (1.2,-0.6);
        \pic at (1.2,-0.6) {figone discard};
        \draw[figone wire] (inputB.east) -- (generator.west |- inputB);
        \draw[figone wire] (inputR.east) -- (generator.west |- inputR);
        \draw[figone wire] (generator.east |- inputB) -- (discP.west |- inputB);
        \draw[figone wire] (generator.east |- inputR) -- (3.55,-1.8);
        \pic at (3.55,-1.8) {figone discard};
        \draw[figone wire] (inputD.east) -- (discP.west |- inputD);

        \foreach \level in {0,-1.2,-2.4}{
            \draw[figone wire] (discP.east |- 0,\level) -- (5.75,\level);
            \pic at (6,\level) {figone readout};
        }
    \end{tikzpicture}
    \par\smallskip
    (a) sampling $p_{\theta,\phi}$
    \end{minipage}\hfill
    \begin{minipage}[t]{0.36\linewidth}
    \centering
    \begin{tikzpicture}[x=1cm,y=1cm,font=\normalsize]
        \path[use as bounding box] (-0.65,-2.9) rectangle (3.8,0.4);

        \draw[figone wire,decorate,decoration={brace,amplitude=3pt,mirror}] (0.35,0.12) -- (0.35,-1.32);
        \node[anchor=east] at (0.16,-0.6) {$\rho_{AB}$};
        \node[anchor=east] (inputD) at (0.55,-2.4) {$\ket{0}_D$};
        \node[figone discriminator,minimum width=1cm,minimum height=2.95cm]
            (discQ) at (2,-1.2) {$V^\dagger(\phi)$};

        \foreach \level in {0,-1.2,-2.4}{
            \draw[figone wire] (0.55,\level) -- (discQ.west |- 0,\level);
            \draw[figone wire] (discQ.east |- 0,\level) -- (3.05,\level);
            \pic at (3.3,\level) {figone readout};
        }
    \end{tikzpicture}
    \par\smallskip
    (b) sampling $q_\phi$
    \end{minipage}
    \caption{Sampling circuits for bipartite-state estimation. Each wire represents a register of qubits. (a) Sampling $p_{\theta,\phi}$ from $\rho_A\otimes\sigma_B(\theta)$. The original $B$ and the purification register $R$ are discarded. (b) Sampling $q_\phi$ from $\rho_{AB}$. Both panels apply the same $V^\dagger(\phi)$ followed by computational-basis measurements, with the measured registers ordered as $A,B,D$ from top to bottom.}
    \label{fig:framework}
\end{figure}
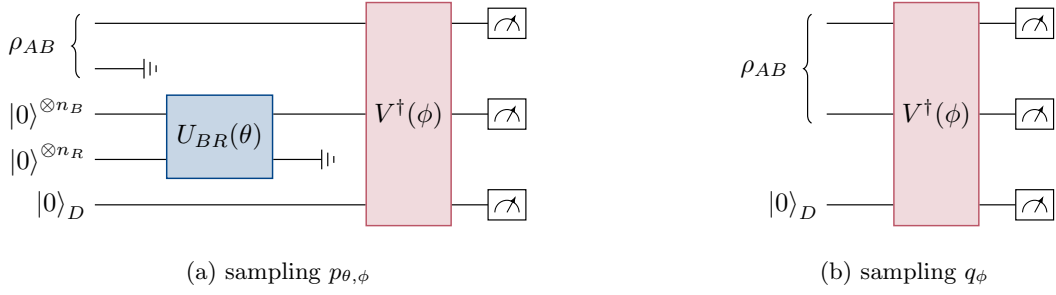

The distributions in Eq.~\eqref{eq:pq} are obtained from independent copies of the unknown state $\rho_{AB}$ using the two sampling circuits shown in Fig.~\ref{fig:framework}. To obtain $p_{\theta,\phi}$, we discard the original subsystem $B$ of a fresh copy of $\rho_{AB}$, leaving its marginal $\rho_A$, and generate $\sigma_B(\theta)$ by applying $U_{BR}(\theta)$ to fresh registers $B$ and $R$. After appending the register $D$, we apply $V^\dagger(\phi)$ to $ABD$ and measure all its qubits in the computational basis. To obtain $q_\phi$, we append $D$ to another copy of $\rho_{AB}$ and apply the same discriminator circuit and measurement. Thus, no tomographic reconstruction of either $\rho_{AB}$ or $\rho_A$ is required.

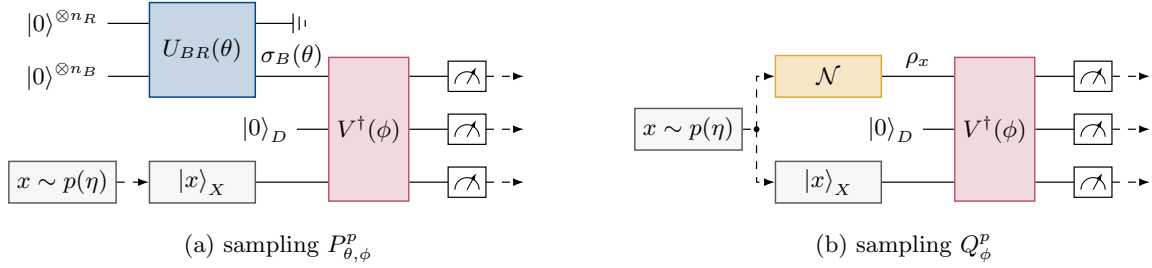
\begin{figure}[!htbp]
\definecolor{figGenerator}{HTML}{004488}
\definecolor{figDiscriminator}{HTML}{BB5566}
\definecolor{figChannel}{HTML}{DDAA33}
\tikzset{
    every node/.style={font=\fontsize{9}{11}\selectfont},
    block/.style={draw, line width=0.55pt, fill=white,
        align=center, inner sep=3pt},
    generator/.style={block,draw=figGenerator,fill=figGenerator!22},
    preparation/.style={block,draw=black!65,fill=black!3,minimum width=1.4cm,minimum height=0.55cm},
    sampler/.style={block,draw=black!65,fill=black!3,minimum width=1.3cm,minimum height=0.55cm},
    discriminator/.style={block,draw=figDiscriminator,fill=figDiscriminator!22,minimum width=1cm,minimum height=1.9cm},
    channel/.style={preparation,draw=figChannel,fill=figChannel!28},
    quantum/.style={line width=0.45pt},
    classical/.style={dashed, line width=0.45pt, -{Latex[length=1.4mm]}},
    pics/figtwo discard/.style={code={
        \draw[quantum] (0,-0.13) -- (0,0.13);
        \draw[quantum] (0.08,-0.09) -- (0.08,0.09);
        \draw[quantum] (0.16,-0.045) -- (0.16,0.045);
    }},
    pics/readout/.style={code={
        \draw[line width=0.45pt,fill=white] (-0.25,-0.2) rectangle (0.25,0.2);
        \draw[line width=0.4pt] (-0.16,-0.1)
            arc[start angle=180,end angle=0,radius=0.16];
        \draw[line width=0.4pt,-{Latex[length=1mm,width=0.7mm]}]
            (0,-0.1) -- (0.12,0.12);
    }}
}

\centering
\begin{minipage}[t]{0.48\linewidth}
\centering
\begin{tikzpicture}[x=1cm,y=1cm]
    \path[use as bounding box] (0,-1.85) rectangle (7.2,1.1);

    \node (inputR) at (0.75,0.7) {$\ket{0}^{\otimes n_R}$};
    \node (inputB) at (0.75,0) {$\ket{0}^{\otimes n_B}$};
    \node[generator,minimum width=1.4cm,minimum height=1.25cm]
        (generator) at (2.6,0.35) {$U_{BR}(\theta)$};
    \node[sampler] (sampleP) at (0.75,-1.4) {$x\sim p(\eta)$};
    \node[preparation] (recordP) at (2.6,-1.4) {$\ket{x}_X$};
    \node[discriminator] (discP) at (4.8,-0.7) {$V^\dagger(\phi)$};
    \node[anchor=east] (ancillaP) at (3.85,-0.7) {$\ket{0}_D$};

    \draw[quantum] (inputR.east) -- (generator.west |- inputR);
    \draw[quantum] (inputB.east) -- (generator.west |- inputB);
    \draw[quantum] (generator.east |- inputR) -- (3.8,0.7);
    \pic at (3.8,0.7) {figtwo discard};
    \draw[quantum] (generator.east |- inputB) -- node[above] {$\sigma_B(\theta)$} (discP.west |- inputB);
    \draw[classical] (sampleP.east) -- (recordP.west);
    \draw[quantum] (recordP.east) -- (discP.west |- recordP);
    \draw[quantum] (ancillaP.east) -- (discP.west |- ancillaP);
    \foreach \level in {0,-0.7,-1.4}{
        \draw[quantum] (discP.east |- 0,\level) -- (5.85,\level);
        \pic at (6.1,\level) {readout};
        \draw[classical] (6.35,\level) -- (6.85,\level);
    }
\end{tikzpicture}
\par\smallskip
{\small (a) sampling $P_{\theta,\phi}^{p}$}
\end{minipage}\hfill
\begin{minipage}[t]{0.48\linewidth}
\centering
\begin{tikzpicture}[x=1cm,y=1cm]
    \path[use as bounding box] (0,-1.85) rectangle (7.2,1.1);

    \node[sampler] (sampleQ) at (0.75,-0.7) {$x\sim p(\eta)$};
    \node[preparation] (recordQ) at (2.6,-1.4) {$\ket{x}_X$};
    \node[channel] (channel) at (2.6,0) {$\mathcal N$};
    \node[discriminator] (discQ) at (4.8,-0.7) {$V^\dagger(\phi)$};
    \node[anchor=east] (ancillaQ) at (3.85,-0.7) {$\ket{0}_D$};

    \coordinate (labelQ) at (1.65,-0.7);
    \draw[dashed,line width=0.45pt] (sampleQ.east) -- (labelQ);
    \fill (labelQ) circle[radius=1.1pt];
    \draw[classical] (labelQ) |- (recordQ.west);
    \draw[classical] (labelQ) |- (channel.west);
    \draw[quantum] (recordQ.east) -- (discQ.west |- recordQ);
    \draw[quantum] (channel.east) -- node[above] {$\rho_x$} (discQ.west |- channel);
    \draw[quantum] (ancillaQ.east) -- (discQ.west |- ancillaQ);
    \foreach \level in {0,-0.7,-1.4}{
        \draw[quantum] (discQ.east |- 0,\level) -- (5.85,\level);
        \pic at (6.1,\level) {readout};
        \draw[classical] (6.35,\level) -- (6.85,\level);
    }
\end{tikzpicture}
\par\smallskip
{\small (b) sampling $Q_\phi^{p}$}
\end{minipage}
\caption{Sampling circuits of the estimation protocol for CQ channels at fixed input distribution $p=p(\eta)$.
Solid lines represent quantum registers; dashed lines carry classical data.
(a) Sampling $P_{\theta,\phi}^{p}$ from $\omega_X^p\otimes\sigma_B(\theta)$,
with $R$ discarded and no channel use.
(b) Sampling $Q_\phi^p$ from $\omega_{XB}^p$,
with the sampled input label recorded in $X$.
Both circuits apply the same $V^\dagger(\phi)$ followed by computational-basis measurements.}
\label{fig:cq-sampling}
\end{figure}

The CQ sampling procedure, shown in Fig.~\ref{fig:cq-sampling}, follows the state preparation and measurement protocol above. However, for state estimation, the device supplies independently prepared copies of $\rho_{AB}$, whereas the CQ channel $\mathcal N$ produces an output $\rho_x$ in response to a chosen input label $x$. At fixed $p$, we prepare the product reference $\omega_X^p\otimes\sigma_B(\theta)$ by sampling an input label $x$ from $p$ and preparing $\ket{x}\!\bra{x}_X\otimes\sigma_B(\theta)$, which requires no channel use. We prepare the labelled output state $\omega_{XB}^p$ by drawing $x\sim p$, supplying it to the channel, and retaining a classical record of $x$ in $X$ together with the output $\rho_x$. Applying the same discriminator measurement to both preparations yields samples from $P_{\theta,\phi}^p$ and $Q_\phi^p$ in Eq.~\eqref{eq:cq-measurement-distributions}, respectively, for estimating $\mathcal L_{\mathrm{CQ}}$.

To optimize the channel input, we parametrize its distribution as $p=p(\eta)\in\Delta(\mathcal X)$, where $\eta$ is a vector of classical parameters. We use the shorthand $\mathcal L_{\mathrm{CQ}}(\eta,\theta,h,\phi):=\mathcal L_{\mathrm{CQ}}(p(\eta);\theta,h,\phi)$. For the chosen variational families, the resulting optimization problem is
\begin{equation}
    \sup_{\eta}\inf_{\theta}\sup_{h,\phi}
    \mathcal L_{\mathrm{CQ}}\left(\eta,\theta,h,\phi\right).
    \label{eq:cq-parametrized-optimization}
\end{equation}

With $N_p$ and $N_q$ shots for the two sampling circuits, respectively, we estimate the loss for a bipartite state as
\begin{equation}
    \widehat{\mathcal{L}}(\theta,h,\phi)
    =
    \frac{1}{N_p}\sum_{m=1}^{N_p}h_{i_m}
    -
    \frac{1}{N_q}\sum_{m=1}^{N_q}e^{h_{j_m}}
    +1,
    \qquad
    i_m\sim p_{\theta,\phi},
    \quad
    j_m\sim q_\phi.
    \label{eq:empirical}
\end{equation}

For a CQ channel, we estimate the loss as
\begin{equation}
    \widehat{\mathcal L}_{\mathrm{CQ}}(\eta,\theta,h,\phi)
    =
    \frac{1}{N_p}\sum_{m=1}^{N_p}h_{i_m}
    -
    \frac{1}{N_q}\sum_{m=1}^{N_q}e^{h_{j_m}}
    +1,
    \qquad
    i_m\sim P_{\theta,\phi}^{p(\eta)},
    \quad
    j_m\sim Q_\phi^{p(\eta)}.
    \label{eq:cq-empirical}
\end{equation}

For fixed parameters, $\widehat{\mathcal{L}}$ and $\widehat{\mathcal L}_{\mathrm{CQ}}$ are  unbiased estimators of $\mathcal{L}$ and $\mathcal L_{\mathrm{CQ}}$, respectively.

For independent shots in the two circuits, the finite-sample uncertainty at fixed parameters is quantified by

\begin{equation}
    \operatorname{Var}\!\left(\widehat{\mathcal L}\right)
    =
    \frac{\operatorname{Var}_{i\sim p_{\theta,\phi}}(h_i)}{N_p}
    +
    \frac{\operatorname{Var}_{i\sim q_\phi}(e^{h_i})}{N_q},
    \qquad
    \operatorname{Var}\!\left(\widehat{\mathcal L}_{\mathrm{CQ}}\right)
    =
    \frac{\operatorname{Var}_{i\sim P_{\theta,\phi}^{p(\eta)}}(h_i)}{N_p}
    +
    \frac{\operatorname{Var}_{i\sim Q_\phi^{p(\eta)}}(e^{h_i})}{N_q}.
    \label{eq:finite-shot-variance}
\end{equation}

In the CQ case, the measurement-outcome distributions include the randomness of both input-label sampling and measurement outcomes. Both standard deviations therefore decrease as $O(N_p^{-1/2}+N_q^{-1/2})$. After training, these expressions also apply when all learned parameters are held fixed and the losses are evaluated using fresh, independent shots.

To optimize the variational losses, we evaluate their gradients with respect to $\theta$, $\phi$, and $h$, together with the classical input parameters $\eta$ for CQ channels. When each circuit parameter $\theta_k$ or $\phi_k$ enters through a single rotation $e^{-itP/2}$, with $t$ denoting that parameter and $P$ a Pauli operator, the parameter-shift rule~\cite{Mitarai2018} yields
\begin{equation}
    \frac{\partial\mathcal{L}}{\partial\theta_k}
    =
    \frac{1}{2}
    \left[
        \mathcal{L}\!\left(\theta_k+\frac{\pi}{2}\right)
        -
        \mathcal{L}\!\left(\theta_k-\frac{\pi}{2}\right)
    \right],
    \label{eq:shift}
\end{equation}
where only the indicated parameter is shifted while all other parameters are held fixed, and likewise for each $\phi_k$. The same rule applies to $\mathcal L_{\mathrm{CQ}}$ with $\eta$ held fixed. For each eigenvalue $h_i$, on the other hand, the gradients follow directly from Eqs.~\eqref{eq:loss} and~\eqref{eq:cq-variational-loss}, since measurement-outcome distributions do not depend on $h$:

\begin{equation}
    \frac{\partial\mathcal{L}}{\partial h_i}
    =
    p_{\theta,\phi}(i)
    -
    e^{h_i}\,q_{\phi}(i),
    \qquad
    \frac{\partial\mathcal L_{\mathrm{CQ}}}{\partial h_i}
    =
    P_{\theta,\phi}^{p(\eta)}(i)
    -
    e^{h_i}Q_\phi^{p(\eta)}(i).
    \label{eq:hgrad}
\end{equation}

For the CQ input parameters $\eta$, the gradient is obtained by differentiating the classical probabilities $p_x(\eta)$. For a differentiable parametrization $p(\eta)$, holding $(\theta,h,\phi)$ fixed gives
\begin{equation}
    \frac{\partial\mathcal L_{\mathrm{CQ}}}{\partial\eta_k}
    =
    \sum_x
    \frac{\partial p_x(\eta)}{\partial\eta_k}
    \left[
        \sum_i h_iP_{\theta,\phi}(i\mid x)
        -
        \sum_i e^{h_i}Q_\phi(i\mid x)
    \right].
    \label{eq:cq-etagrad}
\end{equation}
Here, $P_{\theta,\phi}(i\mid x)$ and $Q_\phi(i\mid x)$ denote the measurement-outcome distributions for the product-reference and labelled-output preparations, respectively, conditioned on the input label $x$. Neither distribution depends on $\eta$.

Training alternates between gradient-ascent updates of $(\phi,h)$ with $\theta$ fixed and gradient-descent updates of $\theta$ with $(\phi,h)$ fixed. For the CQ loss, these inner updates are performed at fixed $\eta$, followed by an ascent update of $\eta$ using the partial derivative above evaluated at the resulting inner parameters. The updated distribution $p(\eta)$ then determines the probabilities for sampling input labels in the next round of inner optimization. The numbers of discriminator and generator updates per training iteration, together with the number of outer iterations for CQ channels, are treated as hyperparameters.

\section{Numerical simulations}

\label{sec:sim}
In this section, we assess the QGAN-based optimization of Sec.~\ref{subsec:implementation} for bipartite states and CQ channels. We track the objective values during training and compare the resulting estimates with independently computed numerical references for $U^M$.

\subsection{Implementation}
\label{subsec:method}

We first prepare a target state $\rho_{AB}$. To this end, we apply a Haar-random unitary $U_{\mathrm{tar}}$ on $ABE$ and prepare
\begin{equation}
    \ket{\psi}_{ABE}
    =
    U_{\mathrm{tar}}\ket{0}^{\otimes(n_A+n_B+n_E)},
    \qquad
    \rho_{AB}
    =
    \Tr_E\!\left[\ket{\psi}\!\bra{\psi}_{ABE}\right],
    \qquad
    n_E=n_A+n_B,
    \label{eq:sim-target}
\end{equation}
so that $E$ purifies the mixed target state $\rho_{AB}$. Accordingly, each copy of $\rho_{AB}$ entering the circuits of Fig.~\ref{fig:framework} is realized by applying $U_{\mathrm{tar}}$ to a fresh register and leaving $E$ unmeasured.

For a CQ channel $\mathcal N:x\mapsto\rho_x$, we similarly prepare each output state $\rho_x$ entering the circuit of Fig.~\ref{fig:cq-sampling}(b) by applying a Haar-random unitary $U_x$ to fresh registers $BE$ and leaving $E$ unmeasured:
\begin{equation}
    \ket{\psi_x}_{BE}
    =
    U_x\ket{0}^{\otimes(n_B+n_E)},
    \qquad
    \rho_x
    =
    \Tr_E\!\left[\ket{\psi_x}\!\bra{\psi_x}_{BE}\right],
    \qquad
    n_E=n_B.
    \label{eq:sim-cq-target}
\end{equation}
The unitaries $U_x$ are sampled independently for the input labels and held fixed while the input distribution is optimized. Combining these outputs with their classical labels according to the input distribution $p$ gives the joint state $\omega_{XB}^{p}$ defined in Eq.~\eqref{eq:cq-labelled-state}. With these environment dimensions, the Haar-random construction can yield full-rank states in both settings. The classical descriptions of $\rho_{AB}$ and the output states $\rho_x$ are available within each run for computing the corresponding benchmark values.

In both settings, the parametrized circuits $U_{BR}(\theta)$ and $V^{\dagger}(\phi)$ are implemented using the hardware-efficient ansatz approach~\cite{Kandala2017}. The specific layer used here is shown in Fig.~\ref{fig:ansatz}. In each layer, independent $R_x$, $R_y$, and $R_z$ rotations are applied to every qubit, followed by a ring of CNOT gates connecting neighboring qubits, including the last and first qubits. The generator uses $L_G$ layers on $BR$, whereas the discriminator uses $L_D$ layers on $ABD$ for bipartite-state estimation and on $XBD$ for CQ channels.

\begin{figure}[t]
    \centering
    \begin{quantikz}[row sep={0.7cm,between origins}, column sep=0.4cm]
        & \gate{R_x} & \gate{R_y} & \gate{R_z} & \ctrl{1} & \qw      & \targ{}   & \qw \\
        & \gate{R_x} & \gate{R_y} & \gate{R_z} & \targ{}  & \ctrl{1} & \qw       & \qw \\
        & \gate{R_x} & \gate{R_y} & \gate{R_z} & \qw      & \targ{}  & \ctrl{-2} & \qw
    \end{quantikz}
    \caption{One layer of the hardware-efficient ansatz implementing
    $U_{BR}(\theta)$ and $V^{\dagger}(\phi)$, shown for three qubits. Each layer
    applies independent $R_x$, $R_y$, and $R_z$ rotations to every
    qubit, followed by a ring of CNOT gates between neighboring
    qubits, with the last qubit controlling the first.}
    \label{fig:ansatz}
\end{figure}
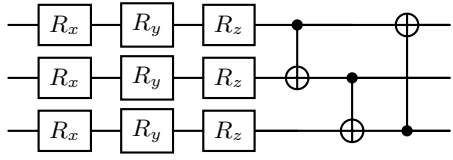

With these circuit parametrizations, the accompanying \texttt{PyTorch}~\cite{Paszke2019} implementation evaluates the full measurement-outcome distributions using a differentiable state-vector simulator. The CQ distributions $P_{\theta,\phi}^{p}$ and $Q_\phi^{p}$ are obtained by computing the conditional distributions $P_{\theta,\phi}(i\mid x)$ and $Q_\phi(i\mid x)$ for each input label $x$ and combining them with weights $p_x$, respectively.

For the numerical optimization, we use the classical Donsker--Varadhan (DV) representation of the Kullback--Leibler divergence~\cite{Belghazi2018} in place of the nonlogarithmic loss. The DV form follows by maximizing the nonlogarithmic loss over a common additive shift of all eigenvalues $h_i$. At fixed circuit parameters, the DV objective for bipartite states takes the form
\begin{equation}
    \mathcal{L}_{\mathrm{DV}}(\theta,h,\phi)
    =
    \sum_{i=0}^{d-1} h_i\,p_{\theta,\phi}(i)
    -
    \log\sum_{i=0}^{d-1} e^{h_i}\,q_{\phi}(i).
    \label{eq:dv}
\end{equation}
The corresponding DV objective for a CQ channel at a fixed input distribution $p$ is obtained by replacing $p_{\theta,\phi}$, $q_\phi$, and $d$ in this expression with $P_{\theta,\phi}^{p}$, $Q_\phi^{p}$, and $d_{\mathrm{CQ}}$, respectively. For fixed measurement-outcome distributions, this logarithmically normalized form and the corresponding nonlogarithmic loss have the same supremum over unrestricted eigenvalues, and hence characterize the same divergence.

To optimize the objective, we initialize the parameters as
\begin{equation}
    \theta_j^{(0)}\sim\mathcal{N}(0,1),
    \qquad
    \phi_k^{(0)}\sim\mathcal{N}\!\left(0,(0.01)^2\right),
    \qquad
    h_i^{(0)}=0,
    \qquad
    p_x^{(0)}=\frac{1}{|\mathcal X|},
    \label{eq:sim-initialization}
\end{equation}
where each rotation angle is sampled independently, and the initialization of $p$ applies only to CQ channels. The gradients of the DV objective are computed by automatic differentiation, and $\theta$, $\phi$, and $h$ are updated with separate Adam optimizers~\cite{Kingma2014}.

We follow the alternating updates of Sec.~\ref{subsec:implementation}. For CQ channels, we take $\eta=p$, update the input probabilities directly by gradient ascent, and project them onto the probability simplex. All learning rates remain at their base values for the first half of training and then decay to 10\% of those values over the second half according to a cosine schedule. After all parameter updates in each iteration, we evaluate the DV objective and record its value.

The register sizes, circuit depths, and training hyperparameters are summarized in Table~\ref{tab:sim-settings}. Each iteration uses one generator update and the listed number of joint $(\phi,h)$ updates; CQ training additionally uses one input-distribution update. These settings are shared by the representative training curves and the corresponding multi-target error plots.

\newpage

\subsection{Results}

The numerical experiments use full-rank target states and full-rank CQ output states, ensuring finite values of $U$ and $U^M$. For the benchmarks, we evaluate $U(A;B)_\rho$ using its closed-form expression in Eq.~\eqref{eq:umlaut-closed-form} and obtain $U(\mathcal N)$ by optimizing the corresponding fixed-input quantity over the input distribution. We additionally compute numerical references for $U^M(A;B)_\rho$ and $U^M(\mathcal N)$ through direct matrix optimization, independently of the variational training and without its circuit restrictions. Each measured reference is reported as the midpoint $\widetilde U^M=(L+B)/2$ of a numerically validated interval $[L,B]$, with $B-L\le10^{-6}$. The relative errors reported below are calculated against these numerical midpoints and do not include the uncertainty of the references. Details of the numerical reference calculations are provided in \cref{app:numerical-reference-calculations}.

The training plots show $U$ alongside $\widetilde U^M$. Comparing the trained values with $\widetilde U^M$ assesses estimation accuracy, whereas the separation of the two reference lines approximates the intrinsic gap $U-U^M$.

\begin{table}[!b]
    \centering
    \caption{Simulation settings for bipartite states and CQ channels. Qubit counts refer to $AB$ and $XB$, respectively. Learning rates are base values before cosine decay. Dashes denote parameters that do not apply.}
    \label{tab:sim-settings}
    \small
    \renewcommand{\arraystretch}{1.08}
    \begin{tabular}{@{}lrrrrrr@{}}
        \toprule
        & \multicolumn{3}{c}{Bipartite states} & \multicolumn{3}{c}{CQ channels} \\
        \cmidrule(lr){2-4}\cmidrule(lr){5-7}
        Qubits & $2$ & $4$ & $6$ & $2$ & $4$ & $6$ \\
        \midrule
        $n_A$ & $1$ & $2$ & $3$ & -- & -- & -- \\
        $|\mathcal X|$ & -- & -- & -- & $2$ & $4$ & $8$ \\
        $n_B$ & $1$ & $2$ & $3$ & $1$ & $2$ & $3$ \\
        $n_E$ & $2$ & $4$ & $6$ & $1$ & $2$ & $3$ \\
        $n_R$ & $1$ & $2$ & $3$ & $1$ & $2$ & $3$ \\
        $n_D$ & $3$ & $3$ & $0$ & $3$ & $3$ & $0$ \\
        \midrule
        Generator layers $L_G$ & $4$ & $4$ & $8$ & $4$ & $8$ & $8$ \\
        Discriminator layers $L_D$ & $20$ & $128$ & $256$ & $20$ & $64$ & $256$ \\
        \midrule
        Training iterations & $1000$ & $1000$ & $2000$ & $1000$ & $1000$ & $4000$ \\
        Joint $(\phi,h)$ updates per iteration & $5$ & $5$ & $20$ & $5$ & $5$ & $10$ \\
        Base learning rate for $\theta$ & $0.01$ & $0.01$ & $0.001$ & $0.01$ & $0.01$ & $0.001$ \\
        Base learning rate for $\phi$ & $0.02$ & $0.04$ & $0.002$ & $0.02$ & $0.02$ & $0.001$ \\
        Base learning rate for $h$ & $0.02$ & $0.01$ & $0.05$ & $0.02$ & $0.02$ & $0.05$ \\
        Base learning rate for $p$ & -- & -- & -- & $0.01$ & $0.01$ & $0.001$ \\
        \bottomrule
    \end{tabular}
\end{table}

Figure~\ref{fig:sim-bipartite-results} tracks training for a fixed bipartite target at each size, with decreasing variation across independent initializations. The two-qubit estimate approaches its reference within tens of iterations; the four-qubit estimate levels off around iterations 850--900 after prolonged fluctuations. The six-qubit overshoot subsides around iterations 200--300, with smaller corrections later. Their separation from $U$ includes the intrinsic gap $U-U^M$, which persists even for exact estimation of $U^M$.

Figure~\ref{fig:sim-bipartite-errors} extends the accuracy assessment to 30 targets per size. The mean absolute relative errors over all targets and runs are approximately $1.6\times10^{-5}\%$, $0.67\%$, and $0.13\%$ for the two-, four-, and six-qubit systems, respectively. Some four-qubit runs show larger deviations, possibly reflecting initialization sensitivity during alternating optimization.

Figure~\ref{fig:sim-cq-results} shows CQ-channel training curves. After initial overshoots, the two- and six-qubit estimates settle near their measured references within roughly 100--150 and 300--500 iterations, respectively. The four-qubit estimate rises more gradually and levels off around iterations 300--400. Run-to-run variation decreases later in training.

The outer optimization adapts the input ensemble when estimating $U^M(\mathcal N)$. In the two-qubit case, the probabilities settle at unequal values with $p_1>p_0$. For the four-qubit channel, the distribution becomes concentrated nearly equally on $x=0$ and $x=2$, while $p_1$ and $p_3$ reach zero. For six qubits, all eight probabilities remain nonzero and continue evolving after the information estimate levels off.

Figure~\ref{fig:sim-cq-errors} summarizes the final relative errors across multiple target CQ channels. For the two-, four-, and six-qubit systems, the mean absolute relative errors over all tested channels and runs are approximately $0.0077\%$, $6.0\%$, and $0.028\%$, respectively. Some four-qubit estimates remain about $10\%$ below their references, possibly reflecting incomplete maximization over the discriminator or input probabilities.

\begin{figure}[p]
    \setlength{\abovecaptionskip}{7pt}
    \centering
    \begin{minipage}[t]{0.48\linewidth}
        \centering
        \small
        \parbox[c][0.25\textheight][c]{\linewidth}{\centering\includegraphics[width=\linewidth]{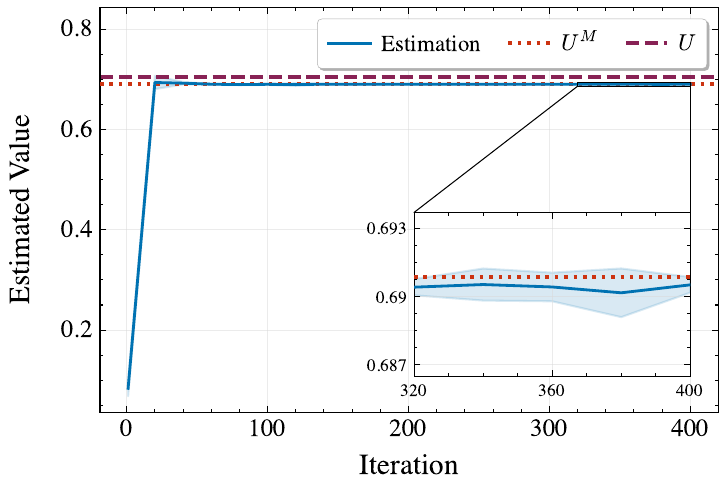}}\par
        (a) 2-qubit system\par
        \smallskip
        \parbox[c][0.25\textheight][c]{\linewidth}{\centering\includegraphics[width=\linewidth]{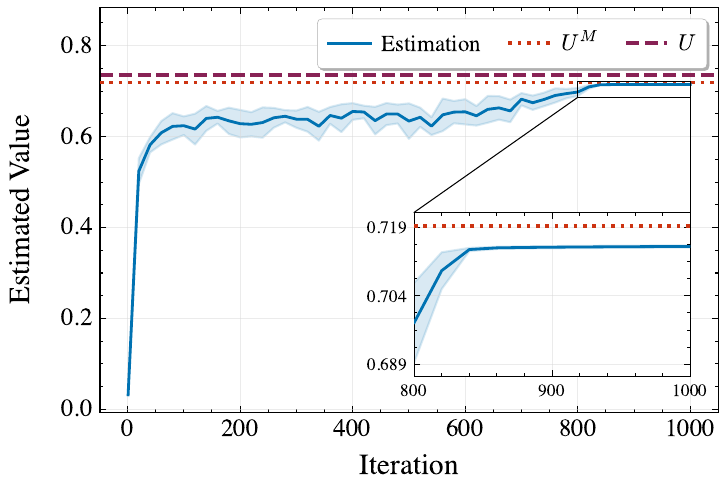}}\par
        (b) 4-qubit system\par
        \smallskip
        \parbox[c][0.25\textheight][c]{\linewidth}{\centering\includegraphics[width=\linewidth]{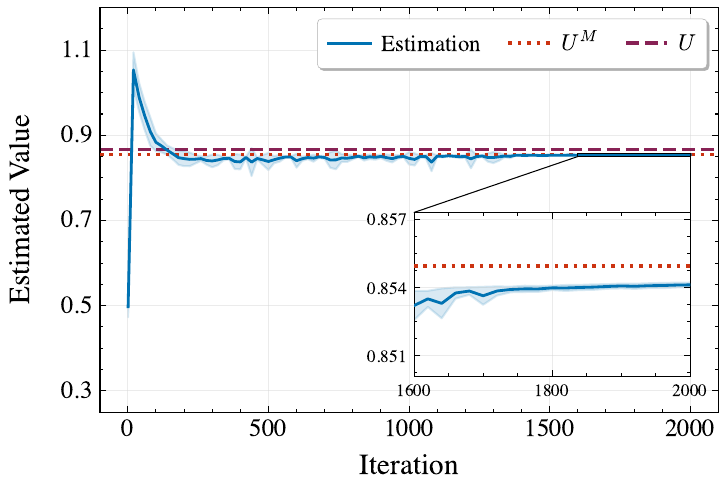}}\par
        (c) 6-qubit system\par
        \caption[Training curves for bipartite-state estimation.]{Training curves for two-, four-, and six-qubit bipartite states in (a)--(c), respectively. Solid curves and shading indicate the mean and $\pm1$ standard deviation over 10 independently initialized runs for a fixed target per panel. Some displayed iteration ranges are shortened to show the approach to a plateau. Training settings are given in Table~\ref{tab:sim-settings}.}
        \label{fig:sim-results}
        \label{fig:sim-bipartite-results}
        \label{fig:sim-two-qubit}
        \label{fig:sim-four-qubit}
    \end{minipage}\hfill
    \begin{minipage}[t]{0.48\linewidth}
        \centering
        \small
        \parbox[c][0.25\textheight][c]{\linewidth}{\centering\includegraphics[width=\linewidth]{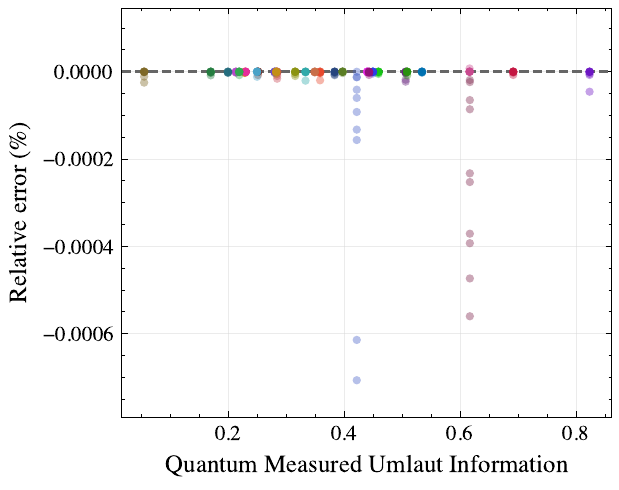}}\par
        (a) 2-qubit system\par
        \smallskip
        \parbox[c][0.25\textheight][c]{\linewidth}{\centering\includegraphics[width=\linewidth]{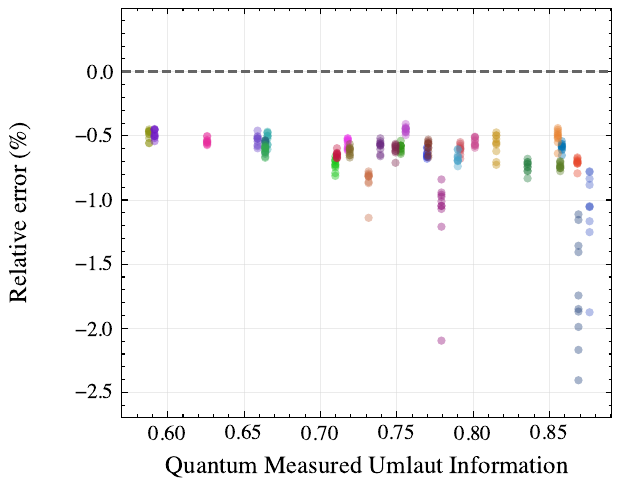}}\par
        (b) 4-qubit system\par
        \smallskip
        \parbox[c][0.25\textheight][c]{\linewidth}{\centering\includegraphics[width=\linewidth]{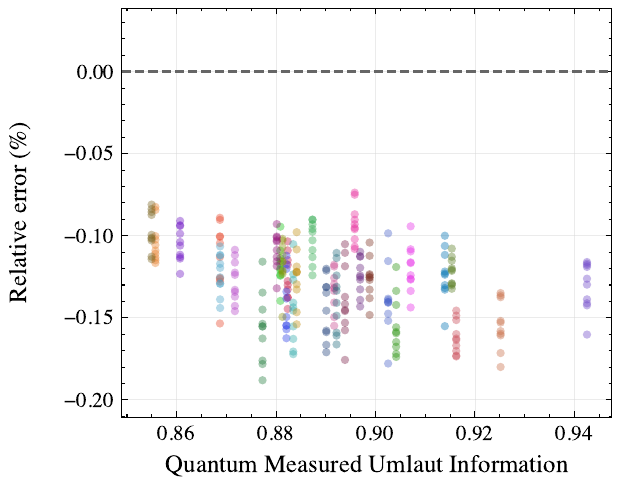}}\par
        (c) 6-qubit system\par
        \caption[Relative errors for bipartite-state estimation.]{Relative errors for bipartite-state estimation. Each panel shows 10 runs for each of 30 targets; each point is one run, and colors distinguish targets. Relative errors of the final estimates $\widehat U^M$ with respect to the numerical references $\widetilde U^M$ are expressed as percentages and plotted against $\widetilde U^M$. Vertical-axis ranges differ across panels. Training settings are given in Table~\ref{tab:sim-settings}.}
        \label{fig:sim-bipartite-errors}
    \end{minipage}
\end{figure}

\begin{figure}[p]
    \setlength{\abovecaptionskip}{2pt}
    \centering
    \begin{minipage}[t]{0.48\linewidth}
        \centering
        \small
        \parbox[c][0.257\textheight][c]{\linewidth}{\centering\includegraphics[width=0.88\linewidth]{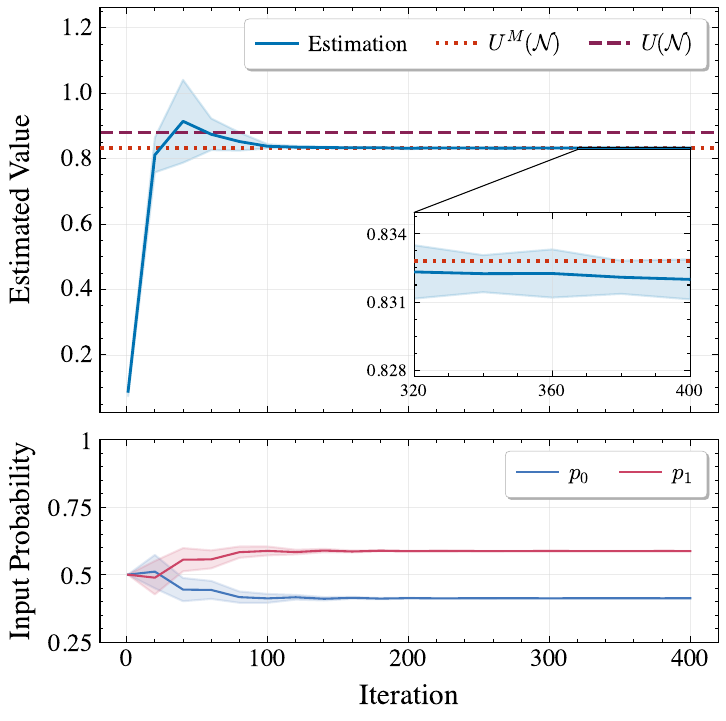}}\par\vspace{3pt}
        (a) 2-qubit system\par
        \vspace{6pt}
        \parbox[c][0.257\textheight][c]{\linewidth}{\centering\includegraphics[width=0.88\linewidth]{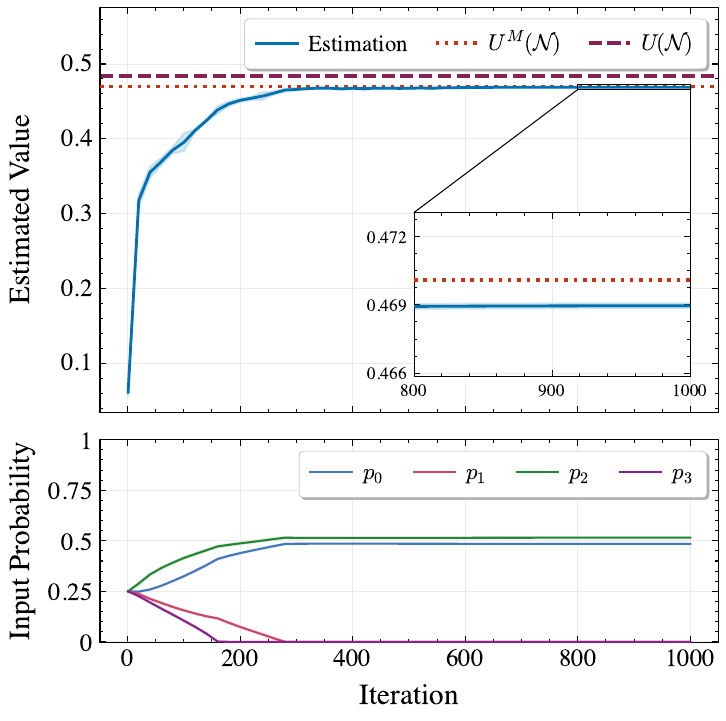}}\par\vspace{3pt}
        (b) 4-qubit system\par
        \vspace{6pt}
        \parbox[c][0.257\textheight][c]{\linewidth}{\centering\includegraphics[width=0.88\linewidth]{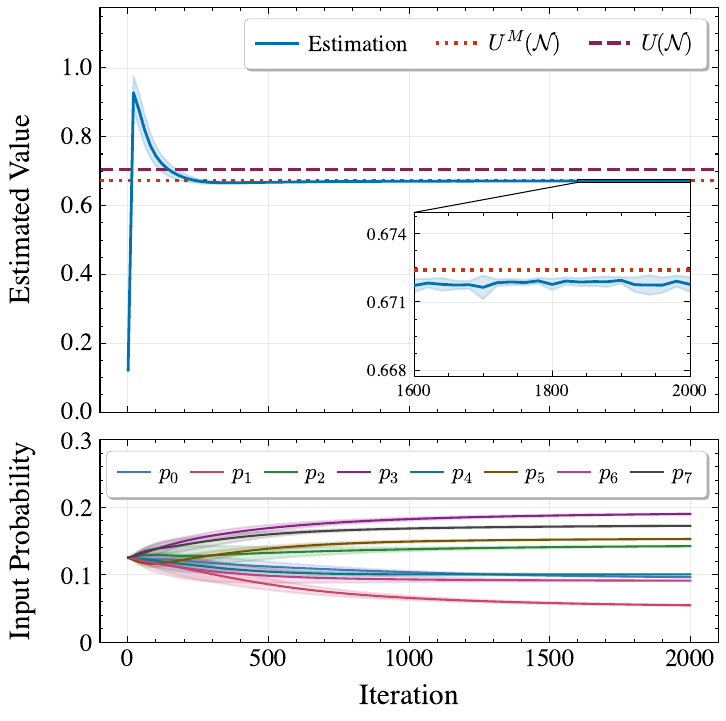}}\par\vspace{3pt}
        (c) 6-qubit system\par
        \caption[Training curves for CQ-channel estimation.]{Training curves for CQ-channel estimation, with two-, four-, and six-qubit joint label--output registers in (a)--(c), respectively. Upper and lower plots show the estimated information and input probabilities $p_x$. Solid curves and shading indicate the mean and $\pm1$ standard deviation over 10 independently initialized runs for a fixed channel per panel. Some displayed iteration ranges are shortened to show the approach to a plateau. Training settings are given in Table~\ref{tab:sim-settings}.}
        \label{fig:sim-cq-results}
    \end{minipage}\hfill
    \begin{minipage}[t]{0.48\linewidth}
        \centering
        \small
        \parbox[c][0.257\textheight][c]{\linewidth}{\centering\includegraphics[width=\linewidth]{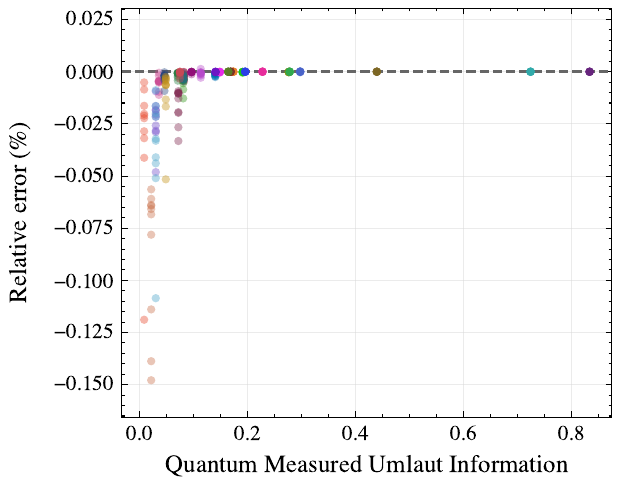}}\par\vspace{3pt}
        (a) 2-qubit system\par
        \vspace{6pt}
        \parbox[c][0.257\textheight][c]{\linewidth}{\centering\includegraphics[width=\linewidth]{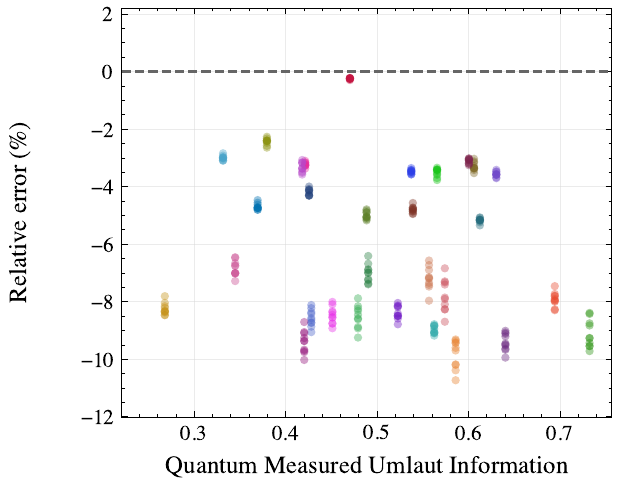}}\par\vspace{3pt}
        (b) 4-qubit system\par
        \vspace{6pt}
        \parbox[c][0.257\textheight][c]{\linewidth}{\centering\includegraphics[width=\linewidth]{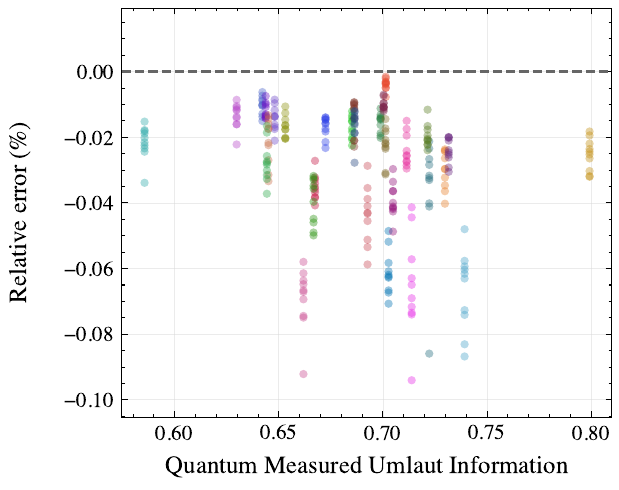}}\par\vspace{3pt}
        (c) 6-qubit system\par
        \caption[Relative errors for CQ-channel estimation.]{Relative errors for CQ-channel estimation; qubit counts refer to $XB$. Each panel shows 10 runs for each of 30 channels; each point is one run, and colors distinguish channels. Relative errors of the final channel estimates $\widehat U^M$ with respect to the numerical references $\widetilde U^M$ are expressed as percentages and plotted against $\widetilde U^M$. Vertical-axis ranges differ across panels. Training settings are given in Table~\ref{tab:sim-settings}.}
        \label{fig:sim-cq-errors}
    \end{minipage}
\end{figure}

\clearpage

\section{Concluding remarks}
\label{sec:conclusion}

We have developed a variational framework for estimating measured quantum umlaut information $U^M$ of bipartite states and CQ channels from quantum measurements. For bipartite states, the QGAN generator varies the reference state $\sigma_B$ while keeping $\rho_A$ fixed, and the discriminator optimizes a measurement basis and classical outcome scores. The unrestricted value of this adversarial game is $U^M(A;B)_\rho$. For CQ channels, a common reference state is used for all input labels, and an outer optimization over their probabilities gives $U^M(\mathcal N)$. The resulting objectives can be evaluated from measurement statistics without reconstructing the target state or the channel output states.

The framework was tested using complete measurement-outcome probabilities for mixed bipartite states of two, four, and six qubits and CQ channels whose joint label--output registers have the same total sizes. Comparisons with independently computed numerical references for $U^M$ show that the trained estimates are close to the measured references on average across the tested targets. These results support the numerical feasibility of the coupled generator, discriminator, and input-distribution optimizations.

We also analyzed the theoretical properties of $U^M$ and its relation to the original quantum umlaut information. For bipartite states, $U^M$ retains faithfulness and monotonicity under local channels. In the finite regime, we bound the intrinsic gap $U-U^M$ in terms of the target spectrum. For full-rank states, we obtain a sharper state-dependent bound, characterize when the gap vanishes, and establish local quadratic stability around equality states. The measured benchmarks assess the estimation of $U^M$, while approximation of $U$ additionally depends on the intrinsic gap. This gap comes from the choice of measured relative entropy and persists independently of circuit and optimization errors.

Estimating $U^M(A;B)$ for an unknown bipartite state provides a way to quantify correlations without reconstructing the density matrix. Comparing this quantity before and after a local operation can then reveal how well the operation preserves those correlations. This could be used to characterize quantum memories: a full-rank mixed probe is shared between a retained reference $A$ and the stored system $B$, and estimates are compared across storage times or noise conditions. Experiments storing and retrieving photonic entanglement illustrate the importance of correlation preservation in atomic-ensemble and solid-state memories~\cite{Choi2008,Clausen2011,Saglamyurek2011}. More recent experiments have demonstrated entanglement between quantum memory nodes connected by telecom fibres~\cite{Knaut2024}. In this setting, our framework would implement the comparison by training directly on repeated preparations before and after storage, yielding a diagnostic of total correlations that complements fidelity or entanglement tests.

Tracking this quantity over time can also reveal revivals of correlations. When $A$ is isolated and only $B$ evolves, $U^M(A;B)$ cannot increase if every intermediate evolution is a completely positive and trace-preserving map. A revival exceeding the estimation uncertainty would therefore indicate failure of this divisibility condition. Such a test could be used to diagnose memory effects in open quantum systems, where established approaches track system--ancilla entanglement~\cite{Rivas2010} or quantum mutual information~\cite{Luo2012,Chen2022}, and information-flow diagnostics have also been demonstrated experimentally~\cite{Fanchini2014}. Using mixed probes with finite information, our estimator could extend this correlation-based approach to $U^M$, learning its time dependence from repeated preparations and joint measurements without reconstructing the evolving states.

For an unknown CQ channel, estimating $U^M(\mathcal N)$ quantifies the correlations between its classical input labels and quantum outputs after optimizing the input distribution. The exact quantity vanishes precisely for input-independent replacer channels, so it provides a criterion for input--output dependence. Repeating the estimation could reveal how that dependence changes under drift or noise. This suggests applications to characterizing quantum transmitters and programmable state-preparation devices with finite channel information. Reconfigurable photonic circuits have demonstrated tunable one-qubit mixed-state preparation~\cite{Shadbolt2012}, while optical experiments have simulated single-qubit noise channels~\cite{Lu2017}. Such a CQ implementation would use classical labels to select preparation settings and their corresponding quantum signal states, as in the ensemble description of classical communication through quantum channels~\cite{Holevo1998}.

With controlled estimation errors, the channel estimate could also provide a lower bound on an operational error exponent. For CQ channels, the exact inequality $U^M(\mathcal N)\le U(\mathcal N)$ and the interpretation of $U(\mathcal N)$ as the optimal zero-rate exponent with activated non-signalling assistance establish this connection~\cite{GirardiQuantum2025}. Such bounds could serve as conservative benchmarks for comparing signal ensembles or operating conditions, while optimizing the input probabilities could search for larger measured bounds. This links the estimator to reliable communication at very low transmission rates. Existing work provides error-probability bounds for unassisted CQ channels~\cite{Dalai2013,DalaiWinter2017} and characterizes reliability exponents with activated non-signalling assistance~\cite{Oufkir2026}.

Several limitations of the present approach remain to be addressed. Circuit restrictions and incomplete adversarial optimization can shift the learned objective above or below $U^M$, so the inequality $U^M\le U$ does not by itself certify a bound from a trained estimate. Finite-shot training and device noise also require validation beyond the present simulations. In addition, assigning a separate trainable eigenvalue to every discriminator outcome leads to exponential growth in the spectrum parameters with the number of measured qubits. Future work should combine compact spectral parametrizations with quantitative control of circuit, optimization, and sampling errors. Such guarantees would allow variational estimates from quantum measurements to support certified assessments of correlations and channel reliability.

\section*{Data Availability Statement}
The data and software that support the findings of this study can be found in the following repository: 
\url{https://github.com/beomjoon01/QUIGAN}

\section*{Acknowledgements}
This work was supported by the National Research Foundation of Korea (NRF) through a grant funded by the Ministry of Science and ICT (Grant No. RS-2025-00515537), the Institute for Information and Communications Technology Promotion (IITP) grant funded by the Korean government (MSIP) (Grant Nos. RS-2019-II190003 and RS-2025-02304540), the National Research Council of Science and Technology (NST) (Grant No. GTL25012-401), and the Korea Institute of Science and Technology Information (Grant No. P26028). We acknowledge the Yonsei University Quantum Computing Project Group for providing support and access to the Quantum System One (Eagle Processor), which is operated at Yonsei University.

\bibliography{reference}

\clearpage
\appendix
\section{Proof of Proposition~\ref{prop:mqui-state-gap}}
\label{sec:proof-mqui-state-gap}
\begin{proof}
The closed-form optimizer in Proposition~\ref{prop:umlaut-closed-form} gives
\begin{equation}\label{eq:mqui-witness-partial-trace}
\begin{aligned}
\Tr_A\!\left[(\rho_A\otimes I_B)H^\star\right]
&=-S(\rho_A)I_B+\log\sigma_B^\star-X_B\\
&=U(A;B)_\rho I_B.
\end{aligned}
\end{equation}
Consequently, $\Tr[(\rho_A\otimes\sigma_B)H^\star]=U(A;B)_\rho$ for every density operator $\sigma_B$. Set $K=H^\star-\Gamma_\rho I_{AB}$, so that $\Tr[\rho_{AB}e^K]=1$. The variational expression of measured relative entropy in Lemma~\ref{lem:dm-variational} implies $D_M(\tau\Vert\rho_{AB})\ge\Tr(\tau K)-\Tr(\rho_{AB}e^K)+1$ for every Hermitian $K$. Evaluating this expression at the same $K$ for every $\sigma_B$ yields
\begin{equation}\label{eq:mqui-uniform-witness-bound} D_M(\rho_A\otimes\sigma_B\Vert\rho_{AB})\ge U(A;B)_\rho-\Gamma_\rho. \end{equation}
Minimizing over $\sigma_B$ proves $\Delta_\rho\le\Gamma_\rho$. The other two upper bounds follow from $U^M\ge0$ and the spectral gap bound in Theorem~\ref{thm:mqui-properties}. The Golden--Thompson inequality~\cite{Golden1965,Thompson1965} gives
\begin{equation}\label{eq:mqui-gamma-nonnegative} e^{\Gamma_\rho}=\Tr[\rho_{AB}e^{H^\star}]\ge\Tr e^{\log\rho_{AB}+H^\star}=\Tr\tau_{AB}^\star=1, \end{equation}
which establishes $\Gamma_\rho\ge0$.

If $[\rho_{AB},\tau_{AB}^\star]=0$, then $\rho_{AB}e^{H^\star}=\tau_{AB}^\star$ and hence $\Gamma_\rho=0$, giving $U^M=U$. Conversely, if $U^M=U$, then $U^M\le D_M(\tau_{AB}^\star\Vert\rho_{AB})\le D(\tau_{AB}^\star\Vert\rho_{AB})=U$ forces equality between the measured and Umegaki relative entropies at $\tau_{AB}^\star$. Both states are full rank, so the equality condition for measured relative entropy~\cite{BFT2017} implies $[\rho_{AB},\tau_{AB}^\star]=0$.
\end{proof}

\section{Proof of Corollary~\ref{cor:mqui-quadratic-stability}}
\label{sec:proof-mqui-quadratic-stability}
\begin{proof}
Set $g(\varepsilon)=\Gamma_{\rho(\varepsilon)}$. Full rank ensures that the matrix logarithms, the Gibbs expression for the umlaut marginal, and hence $g$ are twice continuously differentiable along this family. Proposition~\ref{prop:mqui-state-gap} gives $g(\varepsilon)\ge0$, and commutation at $\varepsilon=0$ gives $g(0)=0$. Thus $g$ has an interior minimum at zero and $g'(0)=0$.

On any sufficiently small closed interval $[-\delta,\delta]$, Taylor's theorem gives $g(\varepsilon)\le C\varepsilon^2$ with $C=\frac12\max_{|t|\le\delta}|g''(t)|<\infty$. Combining this with $\Delta_{\rho(\varepsilon)}\le g(\varepsilon)$ proves $0\le\Delta_{\rho(\varepsilon)}\le C\varepsilon^2$ for $|\varepsilon|\le\delta$.

The closed-form expression $U(A;B)_\rho=-S(\rho_A)-\log\Tr[e^{X_B}]$ also makes $U(A;B)_{\rho(\varepsilon)}$ continuous. If its value at zero is positive, reducing $\delta$ ensures $U(A;B)_{\rho(\varepsilon)}\ge U(A;B)_{\rho(0)}/2$. Dividing the absolute gap bound by this positive lower bound proves $\Delta_{\rho(\varepsilon)}/U(A;B)_{\rho(\varepsilon)}\le[2C/U(A;B)_{\rho(0)}]\varepsilon^2$ for $|\varepsilon|\le\delta$.
\end{proof}

\section{Numerical reference calculations}
\label[appendix]{app:numerical-reference-calculations}
This appendix describes the numerical evaluation of $U$ and $U^M$ for bipartite states and CQ channels from a given bipartite density matrix $\rho_{AB}$ and a given family of CQ-channel output states $\{\rho_x\}_{x\in\mathcal X}$, respectively. For the quantities requiring numerical optimization, we specify the optimization procedure and the construction of lower and upper bounds used to assess the accuracy of the resulting reference value.

The reference value of $U(A;B)_\rho$ for a full-rank state is obtained directly from Eq.~\eqref{eq:umlaut-closed-form}. For a CQ channel with full-rank output states $\rho_x\succ0$, define $A_p=\sum_xp_x\log\rho_x$. Minimization over the common reference state at fixed $p$ gives
\begin{equation}
    \begin{aligned}
        f(p)
        &:=\min_{\sigma_B\in\mathcal D(\mathcal H_B)}
        \sum_xp_xD(\sigma_B\Vert\rho_x)
        =-\log\operatorname{Tr}e^{A_p},\\
        U(\mathcal N)
        &=\max_{p\in\Delta(\mathcal X)}f(p).
    \end{aligned}
    \label{eq:reference-cq-u-objective}
\end{equation}

We maximize $f(p)$ over the probability simplex using sequential least-squares programming (SLSQP), starting from the uniform distribution. The minimizing reference state at fixed $p$ and the analytic gradient used in the optimization are
\begin{equation}
    \sigma_p=\frac{e^{A_p}}{\operatorname{Tr}e^{A_p}},
    \qquad
    g_x(p):=\frac{\partial f(p)}{\partial p_x}
    =-\operatorname{Tr}(\sigma_p\log\rho_x).
    \label{eq:reference-cq-u-gradient}
\end{equation}

For any feasible input distribution $p$, lower and upper bounds satisfying $L\le U(\mathcal N)\le B$ are
\begin{equation}
    L=f(p),
    \qquad
    B=f(p)+\max_{x\in\mathcal X}g_x(p)
    -\sum_xp_xg_x(p).
    \label{eq:reference-cq-u-bounds}
\end{equation}
The lower bound follows by evaluating $f$ at a feasible distribution. As the pointwise minimum of affine functions of $p$, $f$ is concave. Maximizing its affine upper bound $f(p)+\sum_xg_x(p)(q_x-p_x)$ over $q\in\Delta(\mathcal X)$ therefore gives $B$.

For the measured state calculation, write $\rho=\rho_{AB}\succ0$ and $\tau=\rho_A\otimes\sigma_B$. Optimizing the scalar shift of $H$ in Lemma~\ref{lem:dm-variational} gives the logarithmic objective
\begin{equation}
    F(\sigma_B,H)
    =\operatorname{Tr}(\tau H)
    -\log\operatorname{Tr}(\rho e^H).
    \label{eq:reference-state-objective}
\end{equation}
We search for stationary points of $F$ using a numerical root solver, with $H$ and $\sigma_B-I_B/d_B$ represented in traceless Hermitian bases, where $d_B=\dim\mathcal H_B$. The search is initialized using the umlaut marginal from Proposition~\ref{prop:umlaut-closed-form} and $H=\log\tau-\log\rho$, with the scalar component of $H$ removed. The pair $(H,\sigma_B)=(0,I_B/d_B)$ provides an alternative initialization.

For each feasible candidate, define
\begin{equation}
    W=\frac{e^H}{\operatorname{Tr}(\rho e^H)},
    \qquad
    K_B=\operatorname{Tr}_A\!\left[(\rho_A\otimes I_B)\log W\right],
    \label{eq:reference-state-witness}
\end{equation}
so that $W\succ0$ and $\operatorname{Tr}(\rho W)=1$. Let
\begin{equation}
    G=\mathrm D\log_W[\tau]
    =\int_0^\infty
    (W+tI)^{-1}\tau(W+tI)^{-1}\,dt
    \label{eq:reference-log-derivative}
\end{equation}
denote the Fr\'echet derivative of the matrix logarithm at $W$, applied to $\tau$. The lower and upper bounds are
\begin{equation}
    \begin{aligned}
        L&=\max\!\left\{0,\lambda_{\min}(K_B)\right\},\\
        B&=\operatorname{Tr}(\tau\log W)
        +\max\!\left\{
            0,\,
            \lambda_{\max}\!\left(\rho^{-1/2}G\rho^{-1/2}\right)
            -\operatorname{Tr}(GW)
        \right\}.
    \end{aligned}
    \label{eq:reference-state-bounds}
\end{equation}
These satisfy $L\le U^M(A;B)_\rho\le B$ for every feasible candidate, independently of whether the stationarity equations have been solved exactly.

The lower bound follows by evaluating the variational expression at fixed $W$, minimizing over $\sigma_B$, and using the nonnegativity of $U^M$. For the upper bound, concavity of $V\mapsto\operatorname{Tr}(\tau\log V)$ gives
\begin{equation}
    \operatorname{Tr}(\tau\log V)
    \le\operatorname{Tr}(\tau\log W)
    +\operatorname{Tr}[G(V-W)].
    \label{eq:reference-log-tangent}
\end{equation}
Under $V\succ0$ and $\operatorname{Tr}(\rho V)=1$, the supremum of $\operatorname{Tr}(GV)$ is $\lambda_{\max}(\rho^{-1/2}G\rho^{-1/2})$. Consequently, $U^M(A;B)_\rho\le D_M(\tau\Vert\rho)\le B$. Further solver calls refine the candidates as needed, and the narrowest computed interval is retained for the reference-value acceptance test.

For a CQ channel with full-rank output states $\rho_x\succ0$, resolving the classical label in the measurement gives
\begin{equation}
    U^M(\mathcal N)
    =\sup_{p\in\Delta(\mathcal X)}
    \min_{\sigma_B\in\mathcal D(\mathcal H_B)}
    \sum_{x\in\mathcal X}p_xD_M(\sigma_B\Vert\rho_x).
    \label{eq:reference-cq-objective}
\end{equation}
At fixed $p$, the logarithmic objective is the weighted sum of Eq.~\eqref{eq:reference-state-objective}, with $(\tau,\rho,H)$ replaced by $(\sigma_B,\rho_x,H_x)$. We use the same matrix parametrization and numerical root solver to search for stationary candidates, with a common density matrix $\sigma_B$ and a separate Hermitian operator $H_x$ for each label. The search is initialized with $\sigma_B=\sigma_p$ and $H_x=\log\sigma_B-\log\rho_x$, with scalar components removed as above. The alternative initialization uses $\sigma_B=I_B/d_B$ and $H_x=0$ for all $x$.

For each feasible candidate, define
\begin{equation}
    W_x=\frac{e^{H_x}}{\operatorname{Tr}(\rho_xe^{H_x})},
    \qquad
    G_x=\mathrm D\log_{W_x}[\sigma_B],
    \qquad
    d_x=\operatorname{Tr}(\sigma_B\log W_x),
    \label{eq:reference-cq-witnesses}
\end{equation}
where the logarithmic derivative is defined in Eq.~\eqref{eq:reference-log-derivative}. Starting from the uniform input distribution, SLSQP optimizes $p$ over the probability simplex using $\sum_xp_xd_x$ as the approximate inner value and the label-wise values $d_x$ to guide the updates. The inner matrix calculation is repeated at each proposed distribution, using preceding candidates as initial guesses.

Applying the tangent bound in Eq.~\eqref{eq:reference-log-tangent} to each output state gives
\begin{equation}
    b_x=d_x+
    \max\!\left\{
        0,\,
        \lambda_{\max}\!\left(\rho_x^{-1/2}G_x\rho_x^{-1/2}\right)
        -\operatorname{Tr}(G_xW_x)
    \right\},
    \label{eq:reference-cq-label-upper}
\end{equation}
so that $D_M(\sigma_B\Vert\rho_x)\le b_x$. The channel lower and upper bounds are
\begin{equation}
    L=\max\!\left\{
        0,\,
        \lambda_{\min}\!\left(\sum_xp_x\log W_x\right)
    \right\},
    \qquad
    B=\max_{x\in\mathcal X}b_x.
    \label{eq:reference-cq-bounds}
\end{equation}
The lower bound follows by fixing the witnesses, minimizing over the common reference state at the selected distribution $p$, and using nonnegativity, as in the measured state calculation. For the upper bound, the same feasible $\sigma_B$ gives
\begin{equation}
    \min_{\sigma'_B\in\mathcal D(\mathcal H_B)}
    \sum_xq_xD_M(\sigma'_B\Vert\rho_x)
    \le\sum_xq_xD_M(\sigma_B\Vert\rho_x)
    \le\sum_xq_xb_x
    \le B
    \label{eq:reference-cq-upper-justification}
\end{equation}
for every $q\in\Delta(\mathcal X)$. Hence $L\le U^M(\mathcal N)\le B$. The maximum defining $B$ includes all labels, including those with $p_x=0$. Accordingly, the interval width $B-L$ assesses the accuracy of the channel reference including the input-distribution optimization.

For each numerically optimized reference, a value is accepted only when the computed bounds satisfy $B-L\le10^{-6}$. The midpoint of the accepted interval is used in the numerical comparisons, with $\widetilde U^M=(L+B)/2$ denoting the measured reference.

\end{document}